\documentclass[aps,twocolumn,showpacs,showkeys]{revtex4}

\usepackage{amsfonts}
\usepackage{amsmath}
\usepackage{amssymb}
\usepackage{amsthm}

\usepackage{graphicx}

\usepackage{enumitem}

\PassOptionsToPackage{caption=false}{subfig}
\usepackage{subfig}

\usepackage{array}
\usepackage{hhline}
\usepackage{calc}

\newcommand{\T}{\rule{0pt}{3.6ex}}
\newcommand{\B}{\rule[-2.2ex]{0pt}{0pt}}
\newcommand{\M}{\rule{0pt}{3.2ex}}

\theoremstyle{plain}
\newtheorem{theorem}{Theorem}[section]
\newtheorem{proposition}{Proposition}[section]
\theoremstyle{definition}
\newtheorem{definition}[theorem]{Definition}
\newtheorem{condition}[theorem]{Condition}
\theoremstyle{remark}
\newtheorem{example}[theorem]{Example}
\newtheorem{lemma}[theorem]{Lemma}

\newlist{ccindep}{enumerate}{1}
\setlist[ccindep]{label=\textbf{CI\arabic*}:, ref=\textbf{CI\arabic*},
leftmargin=*}
\newlist{qcindep}{enumerate}{1}
\setlist[qcindep]{label={\normalfont\textbf{QCI\arabic*:}},
  ref=\textbf{QCI\arabic*}, leftmargin=*}

\newcommand{\0}{Debbie}
\newcommand{\1}{Wanda}
\newcommand{\2}{Theo}

\newcommand{\Hilb}[1][]{\ensuremath{\mathcal{H}_{#1}}}
\newcommand{\Ket}[1]{\ensuremath{\left \vert #1 \right \rangle}}
\newcommand{\Bra}[1]{\ensuremath{\left \langle #1 \right \vert}}
\newcommand{\Tr}[2][]{\ensuremath{\text{Tr}_{#1} \left ( #2 \right )
  }}
\newcommand{\Lin}[1][]{\ensuremath{\mathfrak{L} \left ( \Hilb[#1] \right )}}

\newcommand{\Sprod}{\ensuremath{\star}}

\begin{document}


\title{A Bayesian approach to compatibility, improvement, and pooling
  of quantum states}
\date{October 3, 2011}
\author{M. S. Leifer}
\affiliation{Department of Physics and Astronomy, University College
  London, Gower Street, London WC1E 6BT, United Kingdom}
\email{matt@mattleifer.info}
\homepage{http://mattleifer.info}
\author{Robert W. Spekkens}
\affiliation{Perimeter Institute for
  Theoretical Physics, 31 Caroline St. N., Waterloo, Ontario, Canada,
  N2L 2Y5}
\email{rspekkens@perimeterinstitute.ca}
\homepage{http://www.rwspekkens.com}

\begin{abstract}
  In approaches to quantum theory in which the quantum state is
  regarded as a representation of knowledge, information, or belief,
  two agents can assign different states to the same quantum system.
  This raises two questions: when are such state assignments
  compatible? and how should the state assignments of different agents
  be reconciled?  In this paper, we address these questions from the
  perspective of the recently developed conditional states formalism
  for quantum theory \cite{Leifer2011a}.  Specifically, we derive a
  compatibility criterion proposed by Brun, Finkelstein and Mermin
  from the requirement that, upon acquiring data, agents should update
  their states using a quantum generalization of Bayesian
  conditioning.  We provide two alternative arguments for this
  criterion, based on the objective and subjective Bayesian
  interpretations of probability theory.  We then apply the same
  methodology to the problem of quantum state improvement, i.e.\ how
  to update your state when you learn someone else's state assignment,
  and to quantum state pooling, i.e.\ how to combine the state
  assignments of several agents into a single assignment that
  accurately represents the views of the group.  In particular, we
  derive a pooling rule previously proposed by Spekkens and Wiseman
  under much weaker assumptions than those made in the original
  derivation.  All of our results apply to a much broader class of
  experimental scenarios than have been considered previously in this
  context.
\end{abstract}


\maketitle
\tableofcontents

\section{Introduction}

\label{Intro}

In Bayesian probability theory, probabilities represent an agent's
information, knowledge or beliefs; and hence it is possible for two
agents to assign different probability distributions to one and the
same quantity.  Recently, due in part to the emergence of quantum
information theory, there has been a resurgence of interest in
approaches to quantum theory that view the quantum state in a similar
way \cite{Caves2002a, Fuchs2002, Pitowsky2003, Fuchs2003, Fuchs2003a,
  Appleby2005, Caves2006, Pitowsky2006, Spekkens2007, Bub2007,
  Fuchs2010, Fuchs2010a}, and in such approaches it is possible for
two agents\footnote{Traditionalist quantum physicists may prefer
  ``observers'' to ``agents''.  In our view, the term ``agent'' is
  preferable because, whatever happens when one measures a quantum
  system, it cannot be regarded as a mere passive ``observation''.
  The term ``agent'' also emphasizes the close connection to the
  decision theoretic roots of classical Bayesian probability theory.}
to assign different quantum states to one and the same quantum system
(henceforth, to avoid repetition, the term ``state'' will be used to
refer to either a classical probability distribution or a quantum
state).  One way this can arise is when the agents have access to
differing data about the system.  For example, in the BB84 quantum key
distribution protocol \cite{Bennett1984}, Alice, having prepared the
system herself, would assign one of four pure states to the system,
whereas the best that Bob can do before making his measurement is to
assign a maximally mixed state to the system.  This naturally leads to
the question of when two state assignments are \emph{compatible} with
one another, i.e.\ when can they represent validly differing views on
one and the same system?

The meaning of ``validly differing view'' depends on the
interpretation of quantum theory and, in particular, on the status of
the quantum state within it.  If the quantum state is thought of as
being analogous to a Bayesian probability distribution, then the
meaning of ``validly differing view'' also depends on precisely which
approach to Bayesian probability one is trying to apply to the quantum
case.  In the Jaynes-Cox approach \cite{Jaynes2003, Cox1961},
sometimes called objective Bayesianism, states are taken to represent
objective information or knowledge and, given a particular collection
of known data, there is assumed to be a unique state that a rational
agent ought to assign, often derived from a rule such as the Jaynes
maximum entropy principle.  In contrast, in the de
Finetti-Ramsey-Savage approach \cite{Ramsey1931, Finetti1937,
  Finetti1974, Finetti1975, Savage1972}, often called subjective
Bayesianism, states are taken to represent an agent's subjective
degrees of belief and agents may validly assign different states to
the same system even if they have access to identical data about the
system.  This is due to differing prior state assignments, the roots
of which are taken to be unanalyzable by the subjective Bayesian.

In its modern form, the problem of quantum state compatibility was
first tackled by Brun, Finkelstein and Mermin (BFM) \cite{Mermin2002,
  Mermin2002a, Brun2002, Brun2002a}, although this work was motivated
by earlier concerns of Peierls \cite{Peierls1991, Peierls1991a}.  BFM
provide a compatibility criterion for quantum states on finite
dimensional Hilbert spaces.  Mathematically, the criterion is that two
density operators are compatible if the intersection of their supports
is nontrivial.  In particular, the BFM criterion implies that two
distinct pure states are never compatible, so that if any agent
assigns a pure state to the system then any other agent who wishes to
assign a compatible pure state must assign the same one.  In the
special case of commuting state assignments, it also implies the
classical criterion for compatibility of probability distributions on
finite sample spaces, which is that there must be at least one element
of the sample space that is in the support of both distributions.

To date there have been two types of argument given for requiring the
BFM compatibility criterion: one due to BFM themselves \cite{Brun2002}
(an argument that takes a similar point of view was later developed by
Jacobs \cite{Jacobs2002}) and one due to Caves, Fuchs and Schack (CFS)
\cite{Caves2002}.  Although not explicitly given in Bayesian terms,
the BFM argument has an objective Bayesian flavor in that it assumes
that there is a unique quantum state that all agents would assign to
the system if they had access to all the available data.  On the other
hand, the CFS argument is an attempt to give an explicitly subjective
Bayesian argument for the BFM compatibility criterion.  Both arguments
start from lists of intuitively plausible criteria that state
assignments should obey, but, in our view, a more rigorous approach is
needed in order to correctly generalize the meaning that compatibility
has in the classical case.

Classically, there are two arguments for compatibility depending on
whether one adopts the objective or the subjective approach.  In both
cases, compatibility is defined in terms of the rules that Bayesian
probability theory lays down for making probabilistic inferences, and,
in particular the requirement that, upon learning new data, states
should be updated by Bayesian conditioning.  The reason for demanding
an argument based on a well-defined methodology for inference is that
there are situations in which even a Bayesian would want to update
their state assignment by means other than Bayesian conditioning.  For
example, if you discover some information that is better represented
as a constraint than as the acquisition of new data, such as finding
out the mean energy of the molecules in a gas, then minimization of
relative entropy, rather than Bayesian conditioning, would commonly be
used to update probabilities \cite{Kullback1959, Shore1980}.
Arguments have also been made for applying generalizations of Bayesian
conditioning, e.g.\ Jeffrey conditioning \cite{Jeffrey1983,
  Jeffrey2004}, on the acquisition of new data in certain
circumstances.  It is not clear whether the intuitions used by BFM and
CFS are applicable to all such circumstances and indeed our intuitions
about probabilities and quantum states are not all that reliable in
general.  It is therefore important to be clear about the type of
inference procedures that are being allowed for in any argument for a
compatibility condition.

What is missing from the existing arguments for BFM compatibility is a
specification of precisely what sorts of probabilistic inferences are
valid --- in short, a precise quantum analog of Bayesian conditioning.
We have recently proposed such an analog within the formalism of
conditional quantum states \cite{Leifer2011a}.  This formalism has the
advantage of being more causally neutral than the standard quantum
formalism, by which we mean that Bayesian conditioning is applied in
the same way regardless of how the data is causally related to the
system of interest, e.g.\ the data could be the outcome of a direct
measurement of the system, a variable involved in the preparation of
the system, the outcome of a measurement of a remote system that is
correlated with the system of interest, etc.  This causal neutrality
allows us to develop arguments that are applicable in a broader range
of experimental scenarios --- or more accurately, \emph{causal}
scenarios --- than those obtained within the conventional formalism.

In this article, we derive BFM compatibility from the principled
application of the idea that, upon learning new data, agents should
update their states according to our quantum analogue of Bayesian
conditioning. This leaves no room for other principles of a more ad
hoc nature.  Both objective and subjective Bayesian arguments
are given by first reviewing the corresponding classical compatibility
arguments and then drawing out the parallels to the quantum case using
conditional states.  The BFM-Jacobs and CFS arguments are then
criticized in the light of our results.

Having dealt with the question of how state assignments can differ, we
then turn to the question of how to combine the state assignments of
different agents.  In Bayesian theory, the purpose of states is to
provide a guide to rational decision making via the principle of
maximizing expected utility.  In its usual interpretation, this is a
rule for \emph{individual} decision making that does not take into
account the views of other agents.  This raises two conceptually
distinct problems.

Firstly, decision making should be performed on the basis of all
available relevant evidence.  The fact that another agent assigns a
particular state could be relevant evidence, and may cause you to
change your state assignment, even in the case where both state
assignments are the same.  For example, if both you and I assign the
same high probability to some event, then telling you my state
assignment may cause you to assign an even higher probability if you
believe that my reasons for assigning a high probability are valid and
that they are independent of yours.  Following Herbut
\cite{Herbut2004}, we call updating your state assignment in light of
another agent's state assignment \emph{state improvement}.

Secondly, if two agents do have different state assignments, then they
may have different preferences over the available choices in decision
making scenarios.  In practice, decisions often have to be made as a
group, in which case a preference conflict prevents all the agents in
the group from maximizing their individual expected utilities
simultaneously.  This motivates the need for methods of combining
state assignments into a single assignment that accurately represents
the beliefs, information, or knowledge of the group as a whole.  This
problem is called \emph{state pooling}.

In the classical case, both improvement and pooling have been studied
extensively (see \cite{Genest1986} and \cite{Jacobs1995} for reviews).
From this it is clear that there is no hope of coming up with a
universal rule, applicable to all cases, that is just a simple
functional of the different state assignments. Instead, we offer a
general methodology for combining states, in both the classical and
quantum cases, again based on the application of Bayesian
conditioning.

Learning another agent's state assignment can be thought of as
acquiring new data.  Therefore, given our Bayesian methodology, the
state improvement problem is solved by simply conditioning on this
data.  For state pooling, we adopt the \emph{supra-Bayesian} approach
\cite{Keeney1976}, which requires the agents to put themselves in the
shoes of a neutral decision maker.  Although their ability to do this
is not guaranteed, doing so reduces the pooling problem to an instance
of state improvement, i.e.\ the neutral decision maker's state is
conditioned on all the other agents' state assignments and the result
is used as the pooled state.  As with compatibility, our approach to
these problems is to draw out the parallels to the classical case
using conditional states and to derive our results by a principled
application of Bayesian conditioning.  This is an improvement over earlier
approaches \cite{Brun2002a, Poulin2003, Herbut2004, Jacobs2002,
  Jacobs2005, Spekkens2007}, which use more ad hoc principles.
However, some of the results of these earlier approaches are recovered
within the present approach.  In particular, a pooling rule previously
proposed by Spekkens and Wiseman \cite{Spekkens2007} can be derived
from our method in the special case where the minimal sufficient
statistics for the data collected by different agents satisfy a
condition that is slightly weaker than conditional independence.  This
is an improvement on the original derivation, which only holds for a
more restricted class of scenarios.

The results in this paper can be viewed as a demonstration of the
conceptual power of the conditional states formalism developed in
\cite{Leifer2011a}.  However, two concepts that were not discussed in
\cite{Leifer2011a} are required to develop our approach to the state
improvement and pooling problems.  These are quantum conditional
independence and sufficient statistics.  Conditional independence has
previously been studied in \cite{Leifer2008}, from which we borrow the
required results.  Several definitions of quantum sufficient
statistics have been given in the literature
\cite{Barndorff-Nielsen2003, Jencova2006, Petz2008a}, but they concern
sufficient statistics for a quantum system with respect to a classical
parameter, or sufficient statistics for measurement data with respect
to a preparation variable.  By contrast, here we need sufficient
statistics for classical variables with respect to quantum systems.
Our treatment of this is novel to the best of our knowledge.

\section{Review of the conditional states formalism}

\subsection{Basic concepts}

The conditional states formalism, developed in \cite{Leifer2011a},
treats quantum theory as a generalization of the classical theory of
Bayesian inference. In the quantum generalization, classical variables
become quantum systems, and normalized probability distributions over
those variables become operators
on the Hilbert spaces of the systems that have unit trace but are not always positive. The generalization
is summarized in table~\ref{tbl:Review:Basics}, the elements of which
we now review.  The treatment here is necessarily brief.  A more
detailed development of the formalism and its relation to the
conventional quantum formalism can be found in \cite{Leifer2011a}.

Note that we adopt the convention that classical variables are denoted
by letters towards the end of the alphabet, such as $R, S, T, X, Y$
and $Z$, while quantum systems are denoted by letters near the
beginning of the alphabet, such as $A, B$ and $C$.

\begin{table*}[htb]
  \begin{tabular}{|p{15em}|>{\centering}p{15em}|>{\centering}p{15em}|}
    \hhline{|~|-|-|}
    \multicolumn{1}{p{15em}|}{\T\B}& {\bf Classical} & {\bf Quantum} \tabularnewline
    \hline
    \T State & $P(R)$ & $\tau_A$ \tabularnewline
    \M Joint state & $P(R,S)$ & $\tau_{AB}$ \tabularnewline
    \M\B Marginalization & $P(S) = \sum_R P(R,S)$ & $\tau_B = \Tr[A]{\tau_{AB}}$ \tabularnewline
    \hline
    \T Conditional state & $P(S|R)$ & $\tau_{B|A}$ \tabularnewline
    \M\B & $\sum_S P(S|R) = 1$ & $\Tr[B]{\tau_{B|A}} = I_A$ \tabularnewline
    \hline
    \T Relation between joint and & $P(R,S) = P(S|R)P(R)$ & $\tau_{AB} = \tau_{B|A} \Sprod \tau_A$ \tabularnewline
    \M\B  conditional states & $P(S|R) = P(R,S)/P(R)$ & $\tau_{B|A} = \tau_{AB} \Sprod
    \tau_A^{-1}$ \tabularnewline
    \hline
    \T\B Bayes' theorem & $P(R|S) = P(S|R)P(R)/P(S)$ & $\tau_{A|B} = \tau_{B|A} \Sprod (\tau_A \tau_B^{-1})$ \tabularnewline
    \hline
    \T\B Belief propagation & $P(S) = \sum_R P(S|R)P(R)$ & $\tau_{B} = \Tr[A]{\tau_{B|A} \tau_A}$ \tabularnewline
    \hline
  \end{tabular}
  \caption{\label{tbl:Review:Basics}Analogies between the classical
    theory of Bayesian inference and the conditional states formalism
    for quantum theory.}
\end{table*}

In the classical theory of Bayesian inference, a joint probability
distribution $P(R,S)$ describes an agent's knowledge, information or
degrees of belief about a pair of random variables $R$ and $S$.  There
is no constraint on the interpretation of what the two variables can
represent.  They may refer to the properties of two distinct physical
systems at a single time, or to the properties of a single system at
two distinct times, or indeed to any pair of physical degrees of
freedom located anywhere in spacetime.  They may even have a
completely abstract interpretation that is independent of physics,
e.g. $R$ could represent acceptance or rejection of the axioms of
Zermelo-Fraenkel set theory and $S$ could be the truth value of the
Reimann hypothesis.  However, given that we are interested in quantum
theory, such abstract interpretations are of less interest to us than
physical ones.  The main point is that the same mathematical object, a
joint probability distribution $P(R,S)$, is used regardless of the
interpretation of the variables in terms of physical degrees of
freedom.

The theory of quantum Bayesian inference aims to achieve a similar
level of independence from physical interpretation.  In particular, we
want to describe inferences about two systems at a fixed time via the
same rules that are used to describe a single system at two times.  As
such, the usual talk of ``systems'' in quantum theory is
inappropriate, as a system is usually thought of as something that
persists in time.  Instead, the basic element of the conditional
states formalism is a \emph{region}.  An \emph{elementary region}
describes what would normally be called a system at a fixed point in
time and a \emph{region} is a collection of elementary regions.  For
example, whilst the input and output of a quantum channel are usually
thought of as the same system in the conventional formalism, they
correspond to two disjoint regions in our terminology.  This gives a
greater symmetry to the case of two systems at a single time, which
also correspond to two disjoint regions.

A region $A$ is assigned a Hilbert space $\Hilb[A]$ and a composite
region $AB$ consisting of two disjoint regions, $A$ and $B$, is
assigned the Hilbert space $\Hilb[AB] = \Hilb[A] \otimes \Hilb[B]$.
The knowledge, information, or beliefs of an agent about $AB$ are
described by a linear operator on $\Hilb[AB]$ (this operator has other
mathematical properties which will be discussed further on).  This
operator is called the \emph{joint state} and, for the moment, we
denote it by $\tau_{AB}$.  Ideally, one would like this framework to
handle any set of regions, regardless of where they are situated in
spacetime, but unfortunately the formalism developed in
\cite{Leifer2011a} is not quite up to the task.  For instance, it is
currently unclear how to represent degrees of belief about three
regions that describe a system at three distinct times.

In a classical theory of Bayesian inference, one also has the freedom
to conditionalize upon any set of variables, regardless of the
spatio-temporal relations that hold among them, or indeed of the
spatio-temporal relations between the conditioning variables and the
conditioned variables.  Therefore, this is an ideal to which a quantum
theory of Bayesian inference should also strive.  Again, the formalism
of \cite{Leifer2011a} does not quite achieve this ideal.  For instance,
this framework cannot currently deal with pre- and post-selection, for
which the conditioning regions straddle the conditioned system in
time.

Whilst these sorts of consideration limit the scope of our results, we
are still able to treat a wide variety of causal scenarios including
all those that have been previously discussed in the literature on
compatibility, improvement, and pooling.  We begin by providing a
synopsis of the formalism as it has been developed thus
far\footnote{In both Bayesian probability and quantum theory, all
  states depend on the background information available to the agent,
  so all states are really conditional states.  When writing down an
  unconditional probability distribution, this background knowledge is
  assumed implicitly.  However, in the quantum case, the properties of
  a quantum state might well depend on the spatio-temporal relation of
  these implicit conditioning regions to the region of interest.
  Here, we assume that these conditioning regions are related to the
  conditioned region in the standard sort of way.  For instance,
  implicit pre- and post-selection are not permitted.}.

Table~\ref{tbl:Review:Basics} summarizes the basic concepts and
formulas of this framework and defines the terminology that we use for
them.

For an elementary region $A$, the quantum analogue of a normalized
probability distribution is a trace-one operator $\tau_A$ on
$\Hilb[A]$.  For a region $AB$, composed of two disjoint elementary
regions, the analogue of a joint distribution $P(R,S)$ is an operator
$\tau_{AB}$ on $\Hilb[AB]$.  The marginalization operation $P(S) =
\sum_R P(R,S)$ which corresponds to ignoring $R$, is replaced by the
partial trace operation, $\tau_B = \Tr[A]{\tau_{AB}}$, which
corresponds to ignoring $A$.  The role of the marginal distribution
$P(S)$ is played by the marginal state $\tau_B$.

If $A$ is an elementary region, then $\tau_A$ is also \emph{positive},
and simply corresponds to a conventional density operator on $A$.  To
highlight this fact, we denote it by $\rho_A$ in this case.  The
positivity of marginal states on elementary regions implies that the
joint state $\tau_{AB}$ of a pair of elementary regions must have
positive partial traces (but it need not itself be a positive
operator).

Another key concept in classical probability is a conditional
probability distribution $P(S|R)$.  $P(S|R)$ represents an agent's
degrees of belief about $S$ for each possible value of $R$.  It
satisfies $\sum_s P(S=s|R=r) = 1$ for all $r$ and is related to the
joint probability by $P(S|R)= P(R,S)/P(R)$.  This implies Bayes'
theorem, $P(R|S)= P(S|R)P(R)/P(S)$, which allows conditionals to be
inverted.  Conditional probabilities are critical to probabilistic
inference.  In particular, if you assign the conditional distribution
$P(S|R)$ and your state for $R$ is $P(R)$, then your state for $S$ can
be computed from $P(S) = \sum_R P(S|R) P(R)$.  This map from $P(R)$ to
$P(S)$ is called \emph{belief propagation}.

The quantum analogue of a conditional probability is a
\emph{conditional state} for region $B$ given region $A$.  This is a
linear operator on $\Hilb[AB]$, denoted $\tau_{B|A}$, that satisfies
$\Tr[B]{\tau_{B|A}}= I_A$.  It is related to the joint state by
$\tau_{B|A} = \tau_{AB} \Sprod \tau_A^{-1}$, where the
$\Sprod$-product is defined by
\begin{equation}
  M \Sprod N \equiv N^{1/2} M N^{1/2},
\end{equation}
and we have adopted the convention of dropping identity operators and
tensor products, so that $\tau_{AB} \Sprod \tau_A^{-1}$ is shorthand
for $\tau_{AB} \Sprod (\tau_A^{-1} \otimes I_B) = (\tau_A^{-1/2}
\otimes I_B) \tau_{AB} (\tau_A^{-1/2} \otimes I_B)$. The quantum
analogue of Bayes' theorem, relating $\tau_{B|A}$ and $\tau_{A|B}$, is
$\tau_{A|B} = \tau_{B|A} \Sprod (\tau_A \tau_B^{-1})$.  Conditional
states are the key to inference in this framework.  In particular, if
you assign the conditional state $\tau_{B|A}$ and your state for $A$
is $\tau_A$, then your state for $B$ can be computed from $\tau_B
= \Tr[A]{\tau_{B|A} \tau_A}$, where we have used the cyclic
property of the trace.  This map from $\tau_A$ to
$\tau_B$ is called \emph{quantum belief propagation}.

\subsection{The relevance of causal relations}

The rules of classical Bayesian inference are independent of the
causal relationships between the variables under consideration.  For
instance, the formula for belief propagation from $R$ to $S$ does not
depend on whether $R$ and $S$ represent properties of distinct systems
or of the same system at two different times.  Nonetheless, causal
relations between variables can affect the set of probability
distributions that are regarded as plausible models. For example, if
$T$ is a common cause of $R$ and $S$, then $R$ and $S$ should be
conditionally independent given $T$, i.e.\ any viable probability model
should satisfy $P(R,S|T)=P(R|T)P(S|T)$.

In the quantum case, the situation is similar.  The rules of
inference, such as the formula for belief propagation, do not depend
on the causal relations between the regions under consideration, but
causal relations do affect the set of operators that can describe
joint states.  Indeed, the dependence is stronger in the quantum case
because the kind of operator used depends on the causal relation even
for a \emph{pair} of regions.

Suppose that $A$ and $B$ represent elementary regions.  $A$ and $B$
are \emph{causally} related if there is a direct causal influence from
$A$ to $B$ (for instance, if $A$ and $B$ are the input and the output
of a quantum channel), or if there is an indirect causal influence
through other regions (for instance, there is a sequence of channels
with $A$ as the input to the first and $B$ as the output of the last).
$A$ and $B$ are \emph{acausally} related if there is no such direct or
indirect causal connection between them, for instance, if they
represent two distinct systems at a fixed time.

If $A$ and $B$ are \emph{acausally} related, then their joint state
$\tau_{AB}$ is a positive operator.  It simply corresponds to a
standard density operator for independent systems.  The conditionals
$\tau_{A|B}$ and $\tau_{B|A}$ are then also positive operators.  Given
that $\rho$ is the standard notation for density operators, a joint
state of two acausally related regions is denoted $\rho_{AB}$.
Similarly, the conditional states are denoted $\rho_{A|B}$ and
$\rho_{B|A}$.  This notation is meant to be a reminder of the
mathematical properties of these operators.  We refer to them as
\emph{acausal} (joint and conditional) states.

If $A$ and $B$ are \emph{causally} related, then $\tau_{AB}$ does not have to
be a positive operator, but $\tau_{AB}^{T_A}$ (or equivalently
$\tau_{AB}^{T_B}$) is always positive, where $^{T_A}$ and $^{T_B}$
denote partial transpose operations on $A$ and $B$\footnote{Whilst the
  partial transpose operation depends on a choice of basis, the
  \emph{set} of operators that are positive under partial
  transposition is basis independent.}.  Similarly, $\tau_{A|B}$ and
$\tau_{B|A}$ are not necessarily positive, but they must have positive
partial transpose.  In this case, the operators $\tau_{AB},
\tau_{A|B}$ and $\tau_{B|A}$ are denoted $\varrho_{AB}, \varrho_{A|B}$
and $\varrho_{B|A}$ respectively and we refer to them as \emph{causal}
(joint and conditional) states.  In particular, dynamical evolution
taking $\rho_A$ to $\rho_B$ can be represented as quantum belief propagation using a causal conditional
state $\varrho_{B|A}$, i.e. $\rho_B = \Tr[A]{\varrho_{B|A}\rho_A}$.
If, in the conventional formalism, the dynamics would be described by
a Completely Positive Trace-preserving (CPT) map $\mathcal{E}_{B|A}:
\Lin[A]\rightarrow\Lin[B]$, then the corresponding conditional state
$\varrho_{B|A}$ is the operator on $\Hilb[A]\otimes\Hilb[B]$ that is
Jamio{\l}kowski-isomorphic \cite{Jamiolkowski1972} to
$\mathcal{E}_{B|A}$, that is, $\varrho_{B|A} = \sum_{j,k}
\Ket{j}\Bra{k}_A \otimes \mathcal{E}_{B|A'}\left ( \Ket{k}\Bra{j}_{A'}
\right )$, where $\Hilb[A']$ is isomorphic to $\Hilb[A]$ and
$\{\Ket{j}\}$ is any orthonormal basis for $\Hilb[A]$.

\subsection{Modeling classical variables}

Joint, marginal and conditional classical probability distributions
are special cases of joint, marginal and conditional quantum states.
To see this, note that a random variable $R$, with $d_R$ possible
values, can be associated with a $d_R$ dimensional Hilbert space with
a preferred basis $\left \{ \Ket{r_1}_R, \Ket{r_2}_R, \ldots,
  \Ket{r_{d_R}}_R \right \}$ labeled by the possible values of $R$.
Then, a probability distribution $P(R)$ can be encoded in a density
operator that is diagonal in this basis via $ \tau_R = \sum_r P(R=r)
\Ket{r}\Bra{r}_R$.  Similarly, for two random variables, $R$ and $S$,
we can construct Hilbert spaces and preferred bases for each and
encode a joint distribution $P(R,S)$ in a joint state via $\tau_{RS} =
\sum_{r,s} P(R=r,S=s) \Ket{r}\Bra{r}_R \otimes \Ket{s}\Bra{s}_S$, and
a conditional distribution $P(S|R)$ in a conditional state via
$\tau_{S|R} = \sum_{r,s} P(S=s|R=r) \Ket{r}\Bra{r}_R \otimes
\Ket{s}\Bra{s}_S$.

Because all operators on a given classical region commute, the
$\Sprod$-product reduces to the regular operator product for classical
states, so that the formulas for quantum Bayesian inference reduce to
their classical counterparts.  For instance, the quantum Bayes'
theorem becomes $\tau_{R|S}= \tau_{S|R} \tau_{R} \tau_{S}^{-1}$, which
is equivalent to $P(R|S)=P(S|R)P(R)/P(S)$.

Note that if we adopt the convention that partial transposes on
classical regions are always defined with respect to the preferred
basis, then classical joint and conditional states are invariant under
this operation.  Therefore, classical \emph{causal} states have the
same mathematical properties as classical \emph{acausal}
states.\footnote{Even if we do not adopt the convention of evaluating
  partial transposes in the preferred basis, the sets of acausal and
  causal states are still isomorphic.  If $\left \{ \Ket{r}_R \right
  \}$ is the preferred basis for acausal states then this amounts to
  choosing a preferred basis $\left \{ \Ket{r^*}_R \right \}$ for
  causal states, where $^*$ is complex conjugation in the basis used
  to define the partial transpose.  However, this is an unnecessary
  complication that is avoided by adopting the recommended
  convention.}.  Since the notational distinction between $\rho$ and
$\varrho$ is supposed to act as a reminder of the mathematical
difference between causal and acausal states for pairs of quantum
regions, there is no need to make the distinction for classical
states.  We therefore adopt the convention of denoting classical
states over an arbitrary set of regions by $\rho$, regardless of how
the regions are causally related.

To complete our discussion of the basic objects in the conditional
states formalism, we need to describe how correlations between
classical and quantum regions can be represented.  The classical
variable $X$ is represented by a Hilbert space $\Hilb[X]$ with a
preferred basis, as described above, and the quantum region $A$ is
associated with a Hilbert space $\Hilb[A]$ with no preferred
structure.  The hybrid region $XA$ is assigned the Hilbert space
$\Hilb[XA] = \Hilb[X] \otimes \Hilb[A]$, but in representing
correlated states on this space, we must ensure that the classical
part remains classical.  In particular, this means that there can be
no entanglement between $X$ and $A$, and that the reduced state on $X$
must be diagonal in the preferred basis.  This motivates defining a
\emph{hybrid quantum-classical} operator on $\Hilb[XA]$ to be an
operator of the form $M_{XA} = \sum_x \Ket{x}\Bra{x}_X \otimes
M_{X=x,A}$, where each $M_{X=x,A}$ is an operator on $\Hilb[A]$.  The
operators $M_{X=x,A}$ are called the \emph{components} of $M_{XA}$.

It follows that a hybrid joint state has the form $\tau_{XA} = \sum_x
\Ket{x}\Bra{x}_X \otimes \tau_{X=x,A},$ where each component
$\tau_{X=x,A}$ is an operator on $\Hilb[A]$.  Recall that if $X$ and
$A$ are acausally related, then $\tau_{XA}$ must be positive, while if
$X$ and $A$ are causally related, then $\tau_{XA}^{T_A}$ must be
positive.  However, given the form of a hybrid state, $\tau_{XA}$ is
positive if and only if $\tau_{XA}^{T_A}$ is positive, so the two
conditions are equivalent.  Consequently, causal and acausal states on
hybrid regions correspond to the same set of operators.  Therefore, as
for classical states, $\rho$ is used to denote all hybrid states,
regardless of their causal interpretation.

By calculating the marginal state $\rho_X$ and $\rho_A$ from the
hybrid state $\rho_{AX}$, we can define conditional states as
$\rho_{X|A} = \rho_{AX} \Sprod \rho_A^{-1}$ and $\rho_{A|X} =
\rho_{AX} \Sprod \rho_X^{-1}= \rho_{AX} \rho_X^{-1}$.  In the latter
case, the $\Sprod$-product reduces to the regular operator product
because $X$ is classical.  There are two sorts of conditional states
for hybrid systems corresponding to whether the quantum or the
classical region is on the right of the conditional.  If the
conditioning system is quantum, then the conditional state has the
form $\rho_{X|A} = \sum_x \Ket{x}\Bra{x}_X \otimes \rho_{X=x|A}$ where
$\rho_{X=x|A}$ is positive and $\sum_x \rho_{X=x|A} =I_A$.  It follows
that the set of operators $\{ \rho_{X=x|A} \}$ is a Positive Operator
Valued Measure (POVM) and therefore such conditional states can be
used to represent measurements, a fact that we shall make use of in
\S\ref{Model:Single}.  If the conditioning system is classical, then
the conditional state has the form $\rho_{A|X} = \sum_x \rho_{A|X=x}
\otimes \Ket{x}\Bra{x}_X $ where $\rho_{A|X=x}$ is positive and
$\Tr[A]{\rho_{A|X=x}} =1$ for all $x$. The operators $\{\rho_{A|X=x}
\}$ therefore constitute a set of normalized states on $A$, and can
therefore be used to represent state preparations, a fact that will
also be used in \S\ref{Model:Single}.

\subsection{Bayesian conditioning}

Classically, if you are interested in a random variable $R$, and you
learn that a correlated variable $X$ takes the value $x$, then you
should update your probability distribution for $R$ from the prior,
$P(R)$, to the posterior, $P(R|X=x)$.  This is known as \emph{Bayesian
  Conditioning}.

In the conditional states formalism, whenever there is a hybrid
region, regardless of the causal relationship between the classical
variable $X$ and the quantum region $A$, you can always assign a joint
state $\rho_{XA}$.  When you learn that $X$ takes the value $x$, the
state of the quantum region should be updated from $\rho_A$ to
$\rho_{A|X=x}$.  This is \emph{quantum} Bayesian conditioning.

\subsection{How to read this paper}

This article is mainly concerned with the consequences of conditioning
a quantum region on classical data, so the main objects of interest
are hybrid conditional states with classical conditioning regions.  In
this case the set of operators under consideration does not depend on
the causal relation between the two regions.  However, thus far we
have only considered conditioning a quantum region on a \emph{single}
classical variable.  Suppose instead that you learn the values of
\emph{two} classical variables, $X_1$ and $X_2$, and you want to
update your beliefs about a quantum region $A$.  In this case, there
are some causal scenarios where your beliefs cannot be correctly
represented by a joint state $\rho_{AX_1X_2}$.  In such scenarios, our
results do not apply.

To properly explain the distinction between the types of causal
scenario to which our results apply and those to which they do not
requires delving into the conditional states formalism in more detail.
However, this extra material is not necessary for understanding most
of our results, so the reader who is eager to get to the discussion of
compatibility, improvement and pooling can skip ahead to
\S\ref{Compat}, referring back to \S\ref{Model} as necessary.

The next section covers the required background for understanding the
scope of our results and gives several examples of experimental
scenarios to which our results apply. In particular, all of the causal
scenarios that have been considered to date in the literature on
compatibility, improvement, and pooling are within the scope of our
results.  Indeed, given that all previous results have been derived in
the context of specific causal scenarios, our results represent a
substantial increase in the breadth of applicability, even if they do
not yet cover all conceivable cases.

\section{Modeling experimental scenarios using the conditional states
  formalism}

\label{Model}

Table~\ref{tbl:Model:Trans} translates various concepts and formulas
from the conventional quantum formalism into the language of
conditional states.  These correspondences are described in more
detail in \cite{Leifer2011a}.  The meaning of most of the rows should
be evident from the discussion in the previous section, and the rest
are explained in this section as needed.

\begin{table*}[htb]
  \begin{tabular}{|p{18em}|>{\centering}p{15em}|>{\centering}p{15em}|}
    \hhline{|~|-|-|}
    \multicolumn{1}{p{18em}|}{\T\B} & {\bf Conventional Notation} &
    {\bf Conditional States Formalism} \tabularnewline
    \hline
    \T \hspace{0.5em} {\bf Probability distribution of $X$} & $P(X)$ & $\rho_{X}$ \tabularnewline
    \M\B \hspace{0.5em} {\bf Probability that $X=x$} & $P(X=x)$ & $\rho_{X=x}$ \tabularnewline
    \hline
    \T \hspace{0.5em} {\bf Set of states on $A$} & $\{ \rho^A_x \}$ & $\rho_{A|X}$ \tabularnewline
    \M\B \hspace{0.5em} {\bf Individual state on $A$} & $\rho^A_x$ & $\rho_{A|X=x}$ \tabularnewline
    \hline
    \T \hspace{0.5em} {\bf POVM on $A$} & $\{ E^A_x \}$ & $\rho_{X|A}$ \tabularnewline
    \M\B \hspace{0.5em} {\bf Individual effect on $A$} & $E^A_x$ & $\rho_{X=x|A}$ \tabularnewline
    \hline
    \T\B \hspace{0.5em} {\bf Channel from $A$ to $B$} & $\mathcal{E}_{B|A}$ & $\varrho_{B|A}$ \tabularnewline
    \hline
    \M \hspace{0.5em} {\bf Instrument} & $\{ \mathcal{E}^{B|A}_x \}$ & $\varrho_{XB|A}$ \tabularnewline
    \M\B \hspace{0.5em} {\bf Individual Operation} & $\mathcal{E}^{B|A}_x$ & $\varrho_{X=x,B|A}$ \tabularnewline
    \hhline{|=|=|=|}
    \T\B \hspace{0.5em} {\bf The Born rule}  & $\forall x: P(X=x) = \Tr[A]{E^A_x \rho_A}$ & $\rho_X =
    \Tr[A]{\rho_{X|A} \rho_A}$ \tabularnewline
    \hline
    \T\B \hspace{0.5em} {\bf Ensemble averaging} & $\rho_A = \sum_x P(X=x) \rho^A_x$ & $\rho_A = \Tr[X]{\rho_{A|X}
      \rho_X}$ \tabularnewline
    \hline
    \T\B \hspace{0.5em} {\bf Action of a quantum channel} &  $\rho_B = \mathcal{E}_{B|A} \left ( \rho_A \right )$ & $\rho_B =
    \Tr[A]{\varrho_{B|A} \rho_A}$ \tabularnewline
    \hline
    \T\B \hspace{0.5em} {\bf Composition of channels} & $\mathcal{E}_{C|A}
    = \mathcal{E}_{C|B} \circ \mathcal{E}_{B|A}$ & $\varrho_{C|A} =
    \Tr[B]{\varrho_{C|B}\varrho_{B|A}}$ \tabularnewline
    \hline
    \T\B \hspace{0.5em} {\bf State update rule} & $\forall x:P(X=x) \rho^B_x = \mathcal{E}^{B|A}_x \left ( \rho_A \right )$ &
    $\rho_{XB} = \Tr[A]{\varrho_{XB|A} \rho_A}$ \tabularnewline
    \hline
  \end{tabular}
  \caption{\label{tbl:Model:Trans}Translation of concepts and equations from conventional notation to the conditional states formalism.}
\end{table*}

We begin by showing how conditioning a quantum region on a single
classical variable works in several different experimental scenarios.
This is necessary background knowledge for considering the more
relevant scenarios involving conditioning on a pair of variables. The
different experiments correspond to different causal structures, which
are illustrated by directed acyclic graphs.

\subsection{Conditioning on a single classical variable}

\label{Model:Single}

In this section, the quantum region we are interested in making
inferences about is always denoted $B$ and the classical variable on
which the inference is based is denoted $X$.

\begin{figure}
  \subfloat[][]{\label{fig:Model:Prep}
    \quad \includegraphics[scale=0.4]{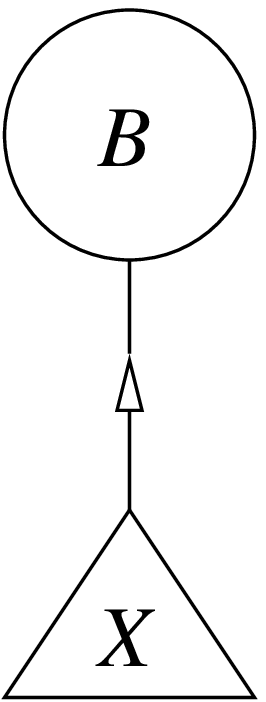} \quad}
  \subfloat[][]{\label{fig:Model:Remote}
    \quad \includegraphics[scale=0.4]{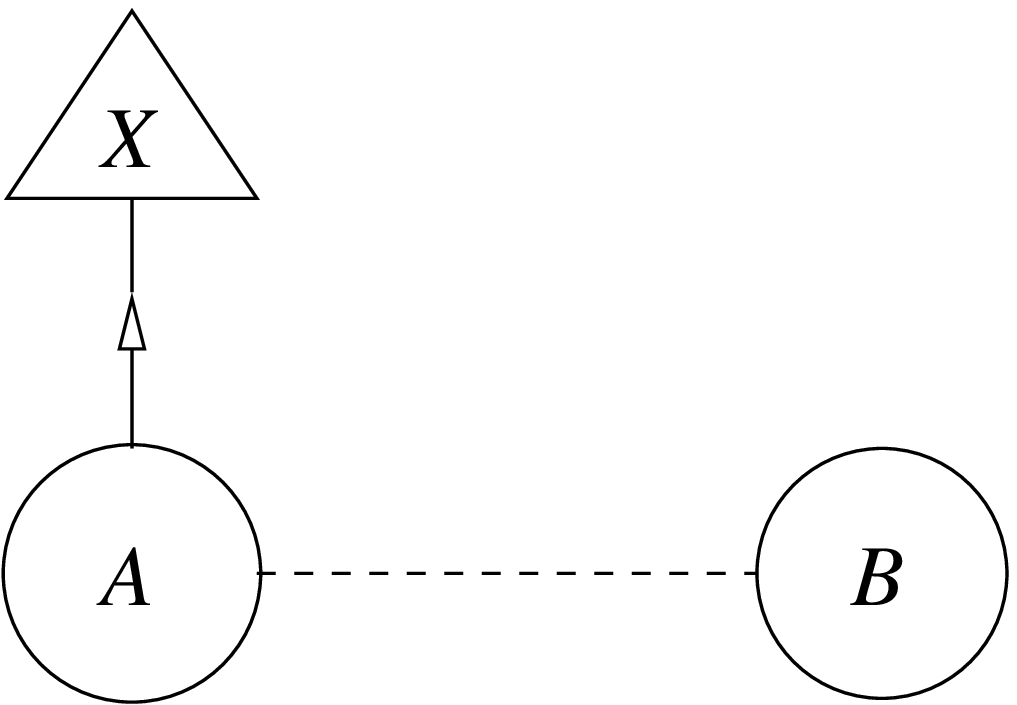} \quad}
  \subfloat[][]{\label{fig:Model:Meas}
    \quad \includegraphics[scale=0.4]{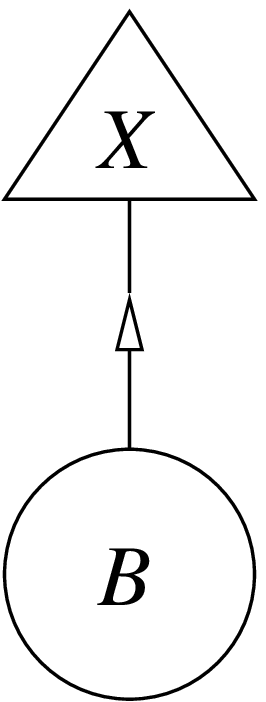} \quad}
  \caption{Quantum-classical hybrid regions with different causal
    relations.  Triangles represent classical variables (as suggested
    by the shape of the probability simplex) and circles represent
    quantum regions (as suggested by the spherical state space of a
    qubit).  \subref{fig:Model:Prep} Preparation procedure: a quantum
    region $B$ is prepared in one of a set of states depending on the
    value of a classical variable $X$ ($B$ is in the causal future of
    $X$).  \subref{fig:Model:Remote} Remote measurement: a
    measurement is made on $A$, which is acausally related to $B$.
    The classical outcome $X$ is then acausally related to $B$.
    \subref{fig:Model:Meas} Measurement: a measurement is made on a
    quantum region $B$ and the classical variable $X$ represents the
    outcome ($X$ is in the causal future of $B$).
    \label{fig:Model:Hybrid}}
\end{figure}

\begin{example}
  Consider the following preparation procedure.  A classical random variable
  $X$ with probability distribution $P(X)$ is generated by flipping
  coins, rolling dice or any other suitable procedure, and then a
  quantum region is prepared in a state $\rho^B_x$ depending on the
  value of $X$ obtained.  This scenario is depicted in
  fig.~\ref{fig:Model:Prep}.  Suppose that, initially, you do not know
  the value of $X$ that was obtained in this procedure.  In the
  conditional states formalism, your beliefs about $X$ are represented
  by a diagonal state $\rho_X$ with components $\rho_{X=x}\equiv
  P(X=x)$.  The set of states prepared is represented by a conditional
  state $\rho_{B|X}$ with components $\rho_{B|X=x}\equiv \rho^B_x$.
  Since the $\Sprod$-product reduces to a regular product for
  classical states, the joint state of $XB$ is $ \rho_{XB} =
  \rho_{B|X}\rho_X$. In terms of components, this is $\rho_{XB} =
  \sum_x P(X=x) \Ket{x}\Bra{x}_X\otimes \rho^B_x$.  It follows that
  $\rho_{XB}$ contains sufficient information to describe an
  \emph{ensemble} of states, i.e.\ a set of states supplemented with a
  probability distribution over them.  Tracing over $X$ gives the
  marginal $\rho_B = \Tr[X]{\rho_{B|X}\rho_X}=\sum_x P(X=x) \rho^B_x$,
  which is easily recognized as the ensemble average state on $B$.

  According to the conventional formalism, upon learning that $X$
  takes the value $x$, you should assign the state that was prepared
  for that particular value of $X$ to $B$, which is just $\rho^B_x$.
  However, since $\rho_{B|X=x} = \rho^B_x$ in the conditional states
  formalism, this update has the form $\rho_B \rightarrow
  \rho_{B|X=x}$, so it is an example of quantum Bayesian conditioning.
  The interpretation of conditioning in this scenario is as an update
  from the ensemble average state to a particular state in the
  ensemble.
\end{example}

\begin{example}
  Suppose that $A$ and $B$ are two acausally related quantum regions
  to which you assign the state $\rho_{AB}$.  The (prior) reduced
  state on $B$ is $\rho_B = \Tr[A]{\rho_{AB}}$.  Now suppose that you
  make a measurement on $A$ with outcome described by the variable $X$
  and that the measurement is associated with a POVM $\{ E^A_x\}$.  In
  the conditional states formalism, the measurement is represented by
  a conditional state $\rho_{X|A}$, where $\rho_{X=x|A}\equiv E^A_x$.
  We are interested in how the state for $B$ gets updated upon
  learning the outcome $x$ of $X$.  This causal scenario is depicted
  in fig.~\ref{fig:Model:Remote}.  This is the scenario that occurs in
  the EPR experiment, or more generally in ``quantum steering''.  The
  update map in this case is sometimes called a ``remote collapse
  rule''.

  In the conditional states formalism, the joint state on $XB$ can be
  determined by belief propagation from $A$ to $X$, i.e. $\rho_{XB} =
  \Tr[A]{\rho_{X|A}\rho_{AB}}$.  The marginal on $X$ gives the outcome
  probabilities for the measurement and is given by $\rho_X =
  \Tr[B]{\rho_{BX}}$. From these, the conditional state $\rho_{B|X}$
  is determined via $\rho_{B|X}=\rho_{BX}\rho_X^{-1}$.  By
  substituting $X=x$ into the expression for $\rho_{B|X}$, we obtain
  $\rho_{B|X=x}$.  This is the state that you should assign to $B$
  when you learn that $X=x$, i.e.\ the update rule for the remote
  region is just Bayesian conditioning $\rho_B \rightarrow
  \rho_{B|X=x}$.  The updated state $\rho_{B|X=x}$ can be expressed in
  terms of the givens in the problem, i.e.\ the state $\rho_{AB}$ and
  the POVM elements $E^A_x$, but this is not especially instructive
  for present purposes. Interested readers can consult
  \cite{Leifer2011a}, where it is shown that this form of Bayesian
  conditioning is precisely the same as the usual remote collapse rule
  in the conventional formalism.
\end{example}

\begin{example}
  Consider the case where $X$ represents the outcome of a direct
  measurement made on $B$ and you want to condition the state of $B$
  on the value of this outcome.  This causal scenario is depicted in
  fig.~\ref{fig:Model:Meas}, and is described by an input state
  $\rho_B$ and a conditional state $\rho_{X|B}$ with components given
  by the POVM that is being measured.  The conditional $\rho_{B|X=x}$
  is then the $X=x$ component of $\rho_{B|X}$, which can be computed
  from an application of Bayes' theorem $\rho_{B|X} = \rho_{X|B}
  \Sprod \left ( \rho_B \rho_X^{-1} \right)$, where $\rho_X =
  \Tr[B]{\rho_{X|B}\rho_B}$.  The operator $\rho_{B|X=x}$ is the state
  that should be assigned to region $B$ upon learning that the outcome
  $X$ takes the value $x$.

  Note that Bayesian conditioning in this case is a kind of
  \emph{retrodiction}: the region being conditioned upon, the outcome
  of the measurement, is to the future of the conditioned region, the
  quantum input to the measurement. This application of Bayesian
  conditioning to retrodiction is discussed in detail in
  \cite{Leifer2011a} and is shown to generate precisely the same
  operational consequences as would be obtained in the conventional
  formalism for retrodiction.
\end{example}

\begin{figure}
  \includegraphics[scale=0.4]{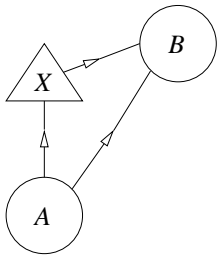}
  \caption{Causal scenario for describing measurement update rules, or
    quantum instruments.  $A$ represents the system before the
    measurement, $X$ is the measurement outcome and $B$ is the system
    after the measurement has been completed.\label{fig:Model:Inst}}
\end{figure}

\begin{example}
  Finally, consider a direct measurement again, but where the region
  of interest is the quantum output of the measurement rather than its
  input.  Let $A$ and $B$ denote the input and output respectively.
  Since these are distinct regions, they must be given distinct labels
  in the conditional states formalism, whereas conventionally they
  would be given the same label as they represent the same system at
  two different times.  The classical variable representing the
  outcome is $X$.  We are interested in how the state of $B$ should be
  updated upon learning the value of $X$.  The relevant causal
  structure is depicted in fig.~\ref{fig:Model:Inst}.  The causal
  arrow from $X$ to $B$ represents the fact that the post-measurement
  state can depend on the measurement outcome in addition to the
  pre-measurement state.

  In general, the rule for determining the state of the region after
  the measurement, given the state of the region before the measurement
  and the outcome, is not uniquely determined by the POVM associated
  with the measurement.  The most general possible rule is
  conventionally represented by a \emph{quantum instrument,} which is
  a set of trace-nonincreasing completely positive maps, $\{
  \mathcal{E}^{B|A}_x\}$.  The operation $\mathcal{E}^{B|A}_x$ maps a
  pre-measurement state $\rho_A$ to the unnormalized post-measurement
  state that should be assigned when the outcome is $x$,
  i.e. $\mathcal{E}^{B|A}_x(\rho_A) = P(X=x)\rho^B_x$, where $P(X=x)$
  is the probability of obtaining outcome $x$ and $\rho^B_x$ is the
  normalized post-measurement state.  This implies that
  if a measurement is associated with a POVM $\{ E^A_x \}$, then
  the quantum instrument must satisfy
  $\Tr[B]{\mathcal{E}^{B|A}_x(\rho_A)}=\Tr[A]{E^A_x \rho_A}$ for all
  input states $\rho_A$.

  It is not too difficult to see how to represent a quantum instrument
  in the conditional states formalism.  First, note that the
  measurement generates an ensemble of states for $B$, i.e.\ for each
  possible outcome $X=x$ there is a probability $P(X=x)$, given by the
  Born rule, and a corresponding state $\rho^B_x$ for $B$, which
  is the state that should be assigned to $B$ when the outcome $X=x$
  occurs.  We have already seen that an ensemble of states can be
  written as a joint state $\rho_{XB}$ of the hybrid region $XB$ via
  $\rho_{XB} = \sum_x P(X=x) \Ket{x}\Bra{x}_X \otimes \rho^B_x$.  What
  is needed then, is a way of determining a joint state $\rho_{XB}$ of
  $XB$, given a state $\rho_A$ of region $A$.  Perhaps unsurprisingly,
  this can be done by specifying a causal conditional state
  $\varrho_{XB|A}$ and using belief propagation to obtain $\rho_{XB} =
  \Tr[A]{\varrho_{XB|A} \rho_A}$.  The POVM that is
  measured by this procedure is given by the components of the
  conditional state $\varrho_{X|A} = \Tr[B]{\varrho_{XB|A}}$.  The
  precise relation between the instrument $\{ \mathcal{E}^{B|A}_x\}$
  and the causal conditional state $\varrho_{XB|A}$ is obtained
  through the Jamio{\l}kowski isomorphism and is described in
  \cite{Leifer2011a}.

  If you assign a prior state $\rho_A$ to the region before the
  measurement, and describe the quantum instrument implementing the
  measurement by $\varrho_{XB|A}$, then the ensemble of output states
  is described by $\rho_{XB}=\Tr[A]{\rho_{XB|A} \rho_A}$.  The
  marginal state $\rho_B = \Tr[X]{\rho_{XB}}$ is then your prior state
  for the output region and $\rho_X = \Tr[B]{\rho_{XB}}$ gives the
  Born rule probabilities for the measurement outcomes.  The states in
  the ensemble, $\rho_{B|X=x}$, can then be computed from the
  conditional $\rho_{B|X}=\rho_{XB} \rho_X^{-1}$.  Upon learning that
  $X=x$, you should update your beliefs about $B$ by Bayesian
  conditioning, i.e.\ by the rule $\rho_{B} \to \rho_{B|X=x}$.

  Note that Bayesian conditioning is \emph{not} a rule that maps your prior
  state about the measurement's input to your posterior state about the
  measurement's output, which would be a map of the form $\rho_A \rightarrow
  \rho_{B|X=x}$.  The projection postulate is an instance of this
  latter kind of update, but it is not an instance of Bayesian
  conditioning.  Bayesian conditioning is a map from prior states to
  posterior states of \emph{one and the same region}. The map
  $\rho_{B} \to \rho_{B|X=x}$, which takes the prior state of the
  measurement's output to the posterior state of the measurement's
  output \emph{is} an instance of quantum Bayesian conditioning.  In
  the conventional formalism it corresponds to a transition from the
  output of a non-selective state-update rule, which you would apply
  when you know that a measurement has occurred but not which outcome
  was obtained, to the output of the corresponding selective
  state-update rule, which applies when you do know the outcome.
\end{example}

\subsection{Conditioning on two classical variables}

\label{Model:Two}

The problems discussed in this paper concern inferences made by
multiple (typically two) agents based on different data.  Thus, we are
interested in conditioning a quantum region on the values
of more than one classical variable, which may or may not be known to
all the agents.

It is convenient to introduce a few more notational conventions to
handle such scenarios.  Since we are using letters to denote regions,
we use numbers to refer to agents.  Given that regions $A$ and $B$ are
prominent in our article, it is confusing to use the usual names Alice
and Bob for our numbered agents, so we refer to agent $1$ as \1 and
agent $2$ as \2. Occasionally, we will refer to a decision-maker, whom
we call \0, and for which we use the number $0$.  A classical variable
that agent $j$ learns during the course of their inference procedure
is denoted $X_j$. The quantum region about which the agents are making
inferences is denoted $B$, and, when making analogies between quantum
theory and classical probability theory, the classical variable
analogous to $B$ is denoted $Y$.  Any other auxiliary quantum regions
involved in setting up the causal scenario are denoted $A$ (or
$A_1,A_2,\ldots$ if there is more than one of them) and auxiliary
classical variables are denoted $Z$ (or $Z_1,Z_2,\ldots$ if there is
more than one of them).

Depending on the causal relations between the classical variables
$X_j$ and an elementary quantum region $B$, it is possible to construct
scenarios in which the available information about the quantum region
cannot be summed up by the assignment of a single state (positive
density operator) to the region.  For example, this is familiar in the
case of pre- and post-selected ensembles, which are described by a
pair of states rather than a single state in the formalism of Aharonov
et.\ al.\ \cite{Aharonov2008}.  Although our results apply to a much
wider variety of causal scenarios than those typically discussed in
the literature on compatibility, improvement, and pooling, we still do
not consider situations in which the region of interest has to be
described by a more exotic object than a single quantum state.  Of
course, a general quantum theory of Bayesian inference \emph{should}
be able to address such scenarios, but that is a topic for future
work.

Mathematically speaking, our results apply whenever the following
condition holds:
\begin{condition} \label{condition:jointdefined}
  The joint region consisting of the quantum region of interest, $B$,
  and all the classical variables involved in the
  inference procedure,  $X_1,X_2\ldots$, can be assigned a joint state
  $\rho_{BX_1X_2\ldots}$ (which may be either an acausal or a causal
  state).
\end{condition}

Consider the case of two classical variables, $X_1$ and $X_2$, and
suppose that a joint state $\rho_{BX_1X_2}$ exists.  From this, one
can compute the reduced states $\rho_B$, $\rho_{X_1}$ and
$\rho_{X_2}$, and the joint states $\rho_{BX_1}$, $\rho_{BX_2}$ and
$\rho_{X_1X_2}$.  From these, one can easily compute the conditional
states $\rho_{B|X_1}$, $\rho_{B|X_2}$ and $\rho_{B|X_1X_2}$.  If \1
learns that $X_1=x_1$ then she updates $\rho_B$ to her posterior state
$\rho_{B|X_1=x_1}$, and if \2 learns that $X_2=x_2$ then he updates
$\rho_B$ to his posterior state $\rho_{B|X_2=x_2}$.  An agent who
learns both outcomes would update to $\rho_{B|X_1=x_1,X_2=x_2}$.  The
existence of the joint state $\rho_{X_1X_2B}$ ensures that all the
posterior states $\rho_{B|X_1=x_1}$, $\rho_{B|X_2=x_2}$ and
$\rho_{B|X_1=x_1,X_2=x_2}$ are well defined.  Similar comments apply
when there are more than two classical variables.

In the remainder of this section, we give several examples of causal
scenarios in which this condition does apply, in order to emphasize
the generality of our results, and we provide some examples where it
does not, to clarify the limitations to their applicability.  All the
examples involve inferences about a quantum region $B$ based on two
classical variables $X_1$ and $X_2$.

\subsubsection{Examples of causal scenarios in which a joint state can be
  assigned}

\begin{figure}
  \subfloat[][]{
    \quad \includegraphics[scale=0.3]{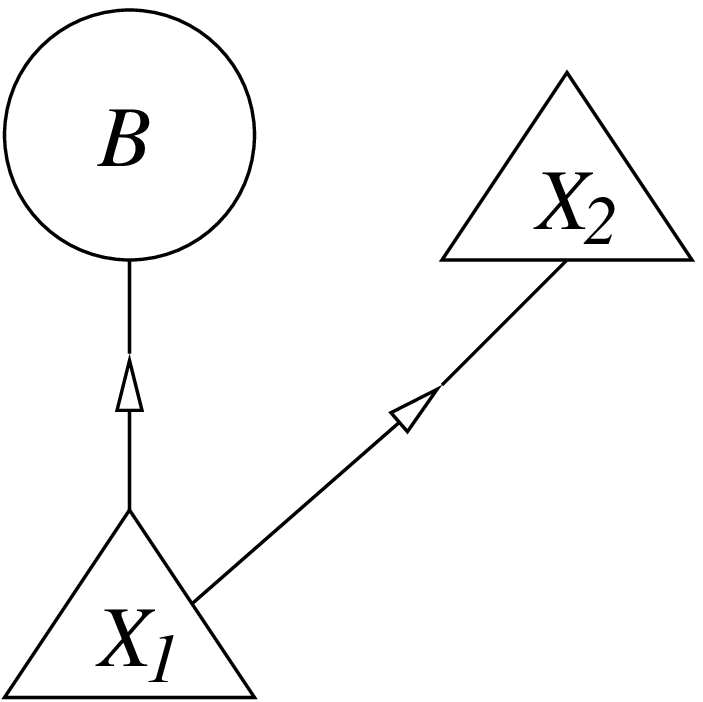} \quad}
  \subfloat[][]{
    \quad \includegraphics[scale=0.3]{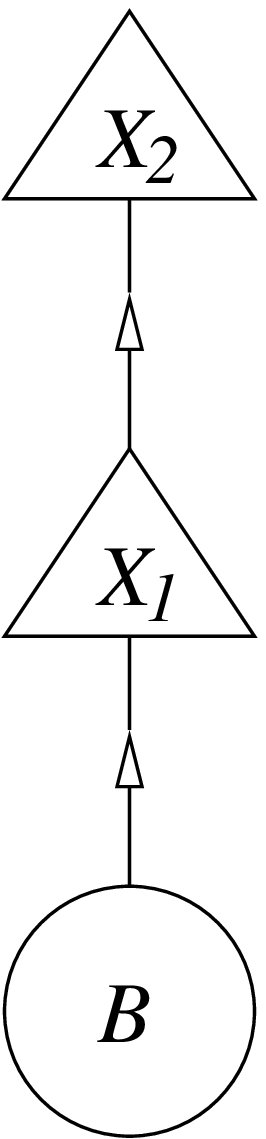} \quad}
  \subfloat[][]{
    \quad \includegraphics[scale=0.3]{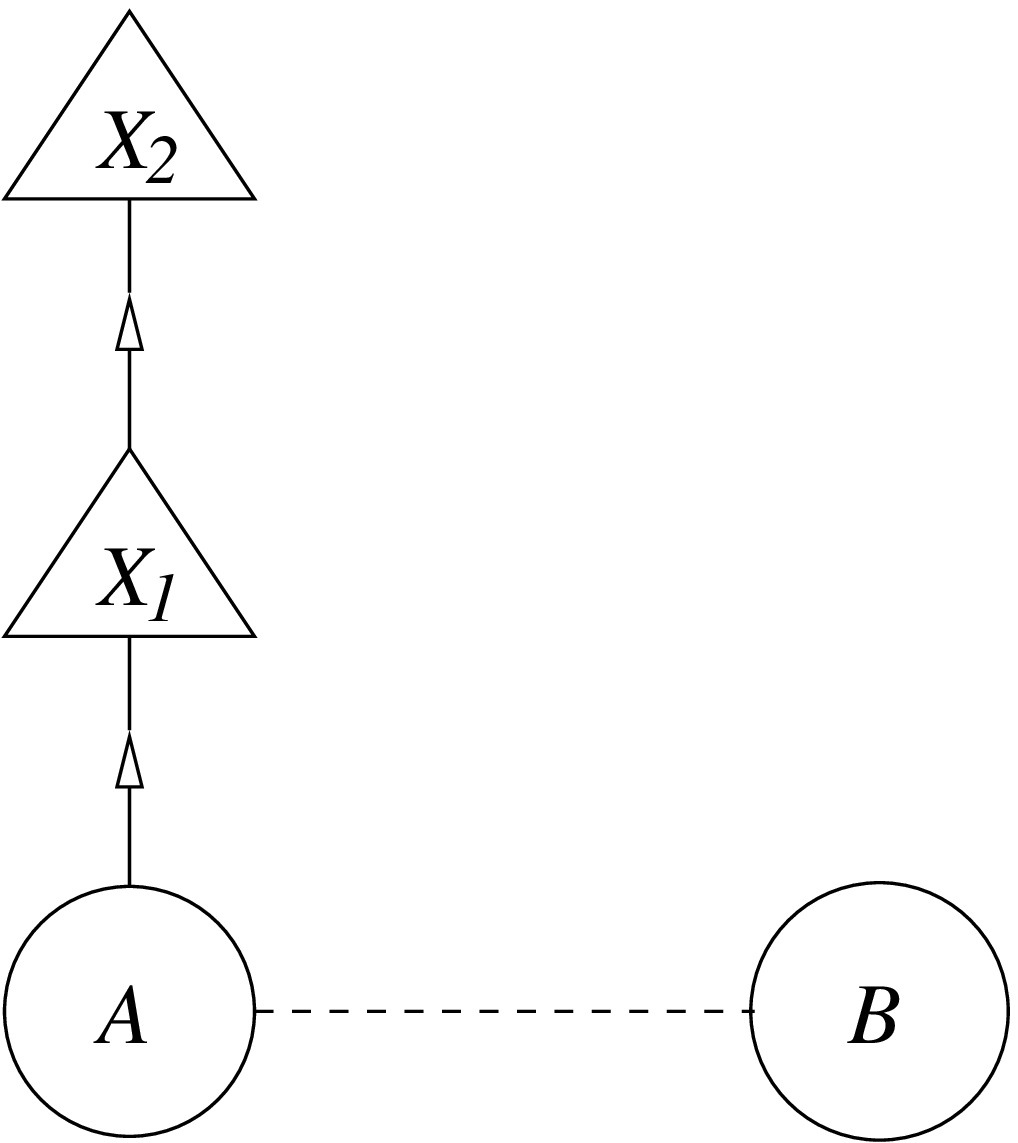} \quad}
  \caption{Introducing an extra classical variable to the causal
    scenarios depicted in fig.~\ref{fig:Model:Hybrid} via
    post-processing.\label{fig:Model:Suff}}
\end{figure}

\begin{example}
  Perhaps the simplest class of causal scenarios in which a joint
  state can be assigned are those in which the second variable $X_2$
  is obtained via a post-processing of the variable $X_1$, i.e. $X_2$
  is obtained from $X_1$ via conditional probabilities $P(X_2|X_1)$,
  or equivalently a classical conditional state $\rho_{X_2|X_1}$.
  Only $X_1$ is directly related to the quantum region $B$ and any
  correlations between $X_2$ and $B$ are mediated by $X_1$.  Examples
  of this sort of causal scenario are depicted in
  fig.~\ref{fig:Model:Suff}.

  In all these scenarios, we already know from \S\ref{Model:Single}
  that $BX_1$ can be assigned a joint state $\rho_{BX_1}$ and then the
  joint state of $BX_1X_2$ is just
  \begin{equation}
    \rho_{BX_1X_2} = \rho_{X_2|X_1} \rho_{BX_1},
  \end{equation}
  so condition \ref{condition:jointdefined} is satisfied.  These
  examples are important because they imply that arbitrary classical
  processing may be performed on a classical variable without changing
  our ability to assign a joint state.  In particular, this is used in
  \S\ref{Suff} where hybrid sufficient statistics are defined as a
  kind of processing of a classical data variable.
\end{example}

\begin{figure}[htb]
  \includegraphics[scale=0.4]{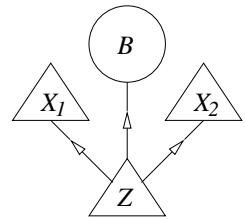}
  \caption{\label{fig:Model:TwoPrep}\1 and \2 learn variables that are
    correlated with a variable used to prepare region $B$.}
\end{figure}

\begin{example}
  Consider a generalization of the preparation scenario depicted in
  fig.~\ref{fig:Model:Prep} to the scenario depicted in
  fig.~\ref{fig:Model:TwoPrep}, which adds two further classical
  variables that depend on the preparation variable.  In this
  scenario, a classical random variable $Z$ is sampled from a
  probability distribution $P(Z)$ and, upon obtaining the outcome
  $Z=z$, a region $B$ is prepared in the state $\rho_{B|Z=z}$.  Some
  data about $Z$ is revealed to the two agents: $X_1$ to \1, and $X_2$
  to \2.  $X_1$ and $X_2$ may be coarse-grainings of $Z$, or they may
  even depend on $Z$ stochastically.  For example, if $Z$ is the
  outcome of a dice roll, then $X_1$ and $X_2$ could both be binary
  variables, with $X_1$ indicating whether $Z$ is odd or even and
  $X_2$ indicating whether it is $\leq 3$.  Generally, the dependence
  of $X_1$ and $X_2$ on $Z$ is given by classical conditional states
  $\rho_{X_1|Z}$ and $\rho_{X_2|Z}$.  A joint state for $X_1X_2B$ can
  be defined in this case via
  \begin{equation}
    \rho_{X_1X_2B} = \Tr[Z]{\rho_{X_1|Z} \rho_{X_2|Z} \rho_{B|Z} \rho_Z},
  \end{equation}
  so, again,condition \ref{condition:jointdefined} is satisfied.
\end{example}

\begin{figure}[htb]
  \includegraphics[scale=0.4]{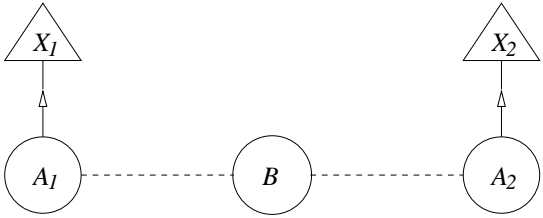}
  \caption{\label{fig:Model:TwoRemote}\1 and \2 learn about $B$ by
    making measurements on two acausally related regions $A_1$ and
    $A_2$.}
\end{figure}

\begin{example}
  Consider the generalization of the remote measurement scenario
  depicted in fig.~\ref{fig:Model:Remote} to \emph{a pair} of remote
  measurements, as depicted in fig.~\ref{fig:Model:TwoRemote}.  This
  scenario is in fact the one that is adopted in much of the
  literature on compatibility and pooling \cite{Brun2002, Brun2002a,
    Spekkens2007}.  The region of interest, $B$, is acausally related
  to two other quantum regions, $A_1$ and $A_2$, so we have a
  tripartite state $\rho_{A_1A_2B}$.  Direct measurements are made on
  $A_1$ and $A_2$, with outcomes $X_1$ and $X_2$ respectively, and
  which are described by the conditional states $\rho_{X_1|A_1}$ and
  $\rho_{X_2|A_2}$ respectively. It is assumed that \1 learns only
  $X_1$ and \2 learns only $X_2$.  In this case, we can define a
  tripartite acausal state by
  \begin{equation}
    \rho_{X_1X_2B} =\Tr[A_1A_2]{\rho_{X_1|A_1}\rho_{X_2|A_2}\rho_{A_1A_2B}}.
  \end{equation}
\end{example}

\begin{figure}[htb]
  \includegraphics[scale=0.4]{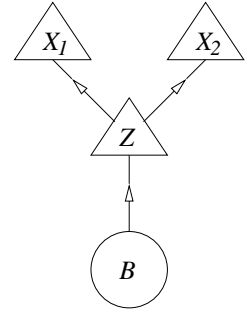}
  \caption{\label{fig:Model:TwoMeas}\1 and \2 learn variables derived
    from a direct measurement made on region $B$.}
\end{figure}

\begin{example}
  Consider a generalization of the direct measurement scenario
  depicted in fig.~\ref{fig:Model:Meas} to the scenario of
  fig.~\ref{fig:Model:TwoMeas}, which introduces two further classical
  variables that depend on the measurement result.  This is similar to
  the second example considered in this section except that, rather
  than $Z$ being used to prepare $B$, it is now obtained by making a
  direct measurement on $B$, described by the conditional state
  $\rho_{Z|B}$.  As before, some information about $Z$ is distributed
  to each agent, specifically, variables $X_1$ and $X_2$ to \1 and \2
  respectively.  The dependence of $X_1$ and $X_2$ on $Z$ is again
  described by conditional states $\rho_{X_1|Z}$ and $\rho_{X_2|Z}$.
  In this case, a joint state $\rho_{X_1X_2B}$ can be defined as
  \begin{equation}
    \rho_{X_1X_2B} = \Tr[Z]{\rho_{X_1|Z}\rho_{X_2|Z} \left [ \rho_{Z|B} \Sprod
        \rho_B \right ]},
  \end{equation}
  and conditioning on values of the classical variables yields states
  that are relevant for retrodiction.
\end{example}

\begin{figure}[htb]
  \includegraphics[scale=0.4]{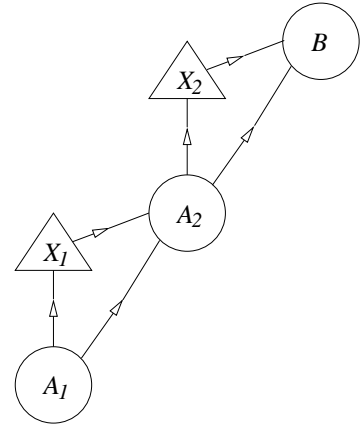}
  \caption{\label{fig:Model:Jacobs}\1 and \2 learn the results of two
    measurements preformed in sequence.}
\end{figure}

\begin{example}
  Consider a generalization of the measurement scenario depicted in
  fig.~\ref{fig:Model:Inst} to a case where a pair of measurements are
  implemented in succession, as depicted in
  fig.~\ref{fig:Model:Jacobs}.  This scenario has been considered in
  the context of compatibility and pooling by Jacobs
  \cite{Jacobs2002}, as discussed in \S\ref{Compat:Jacobs}.  The input
  region of the first measurement is denoted $A_1$. The output of the
  first measurement, which is also the input of the second, is denoted
  $A_2$, and the output of the second measurement, which is the region
  about which \1 and \2 seek to make inferences, is denoted by $B$.
  The classical variables describing the outcomes of the two
  measurements are denoted $X_1$ and $X_2$ respectively, and it is
  assumed that \1 learns $X_1$ while \2 learns $X_2$.

  Suppose that \1 and \2 agree on the input state $\rho_{A_1}$ and on
  the causal conditional states, $\varrho_{X_1A_2|A_1}$ and
  $\varrho_{X_2B|A_2}$, that describe the measurements.  A joint state
  can then be assigned to $X_1X_2B$ via
  \begin{equation}
    \label{eq:Model:Jacobs}
    \rho_{X_1X_2B}=\Tr[A_1A_2]{\varrho_{X_2B|A_2} \varrho_{X_1A_2|A_1}
      \rho_{A_1}}.
  \end{equation}
  The interpretation of eq.~\eqref{eq:Model:Jacobs} is that the two
  consecutive measurements can be thought of as a preparation
  procedure for $B$ that prepares the states in the ensemble
  $\rho_{X_1X_2B}$ depending on the values of $X_1$ and $X_2$.
\end{example}

These examples should serve to give an idea of the type of scenarios
to which our results apply.

\subsubsection{Examples of causal scenarios in which a joint state
  cannot be assigned}

\begin{figure}[htb]
  \includegraphics[scale=0.4]{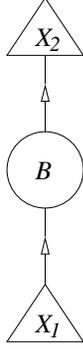}
  \caption{\label{fig:Model:PrepMeas}\1 learns a preparation variables
    and \2 learns a measurement variable.  Learning both variables
    gives a pre- and post-selected ensemble.}
\end{figure}

\begin{example}
  Consider the prepare-and-measure scenario depicted in
  fig.~\ref{fig:Model:PrepMeas}.  Here, $B$ is prepared in a state
  depending on the preparation variable $X_1$ and then $B$ is
  measured, resulting in the outcome $X_2$.  More concretely, consider
  the case where $B$ is a qubit prepared in the $\left \{ \Ket{0}_B,
    \Ket{1}_B \right \}$ basis and measured in the $\left \{
    \Ket{+}_B,\Ket{-}_B \right \}$ basis, where $\Ket{\pm} =
  \frac{1}{\sqrt{2}} \left ( \Ket{0} \pm \Ket{1} \right )$.  Suppose
  that $X_2$ takes the value $X_2=0$ for $\Ket{+}_B$ and $X_2=1$ for
  $\Ket{-}_B$.  Although it is possible to assign joint states to
  $X_1B$ and to $X_2B$, the conditional states that these assignments
  imply are not compatible with any joint state for $X_1X_2B$.

  To see this, note that the joint states for $X_1B$ and $X_2B$ have
  to be of the form
  \begin{multline}
    \rho_{X_1B} = P(X_1=0) \Ket{0}\Bra{0}_{X_1} \otimes
    \Ket{0}\Bra{0}_B \\ + P(X_1=1) \Ket{1}\Bra{1}_{X_1} \otimes
    \Ket{1}\Bra{1}_B
  \end{multline}
  \begin{multline}
    \rho_{X_2B} = P(X_2=0) \Ket{0}\Bra{0}_{X_2} \otimes
    \Ket{+}\Bra{+}_B \\ + P(X_2=1) \Ket{1}\Bra{1}_{X_2} \otimes
    \Ket{-}\Bra{-}_B,
  \end{multline}
  where $P(X_1)$ is the distribution of the preparation variable and
  $P(X_2)$ is the Born rule probability distribution for the outcomes
  of the measurement.

  Then, $\rho_{B|X_1=x_1}$ is a definite state in the $\left \{
    \Ket{0}_B, \Ket{1}_B \right \}$ basis and $\rho_{B|X_2=x_2}$ is a
  definite state in the $\left \{ \Ket{+}_B,\Ket{-}_B \right \}$
  basis.  Any putative $\rho_{B|X_1=x_1,X_2=x_2}$, derived from a
  joint state of all three regions, would then have to have definite
  values for measurements in both the $\left \{ \Ket{0}_B, \Ket{1}_B
  \right \}$ basis and in the $\left \{ \Ket{+}_B,\Ket{-}_B \right \}$
  basis.  There is no state with this property because these are
  complimentary observables.

  Conditioning on both $X_1=x_1$ and $X_2=x_2$ represents a case of
  pre- and post-selection, and, as argued by Aharonov et.\ al.\
  \cite{Aharonov2008}, the concept of a quantum state has to be
  generalized in order to handle such cases.
\end{example}

\begin{figure}[htb]
  \includegraphics[scale=0.4]{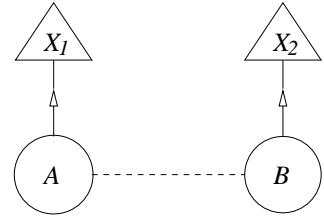}
  \caption{\label{fig:Model:DirectRemote}Learning both the outcome of
    a direct measurement and the outcome of a remote measurement.}
\end{figure}

\begin{example}
  Consider two acausally related quantum regions $A$ and $B$.  Here,
  $B$ is the region of interest, but direct measurements are made on
  both $A$ and $B$, resulting in the classical variables $X_1$ and
  $X_2$ respectively.  This is depicted in
  fig.~\ref{fig:Model:DirectRemote}.  Formally, this is very similar
  to pre- and post-selection and a joint state of $X_1X_2B$ is ruled
  out for similar reasons.

  Suppose that $A$ and $B$ are qubits and that $\rho_{AB} =
  \Ket{\Psi^-}\Bra{\Psi^-}_{AB}$ is a singlet state, where
  $\Ket{\Psi^-}_{AB} = \frac{1}{\sqrt{2}} \left ( \Ket{01} - \Ket{10}
  \right )_{AB}$.  If $X_1$ is the result of a measurement of $A$ in
  the $\left \{ \Ket{0}_A, \Ket{1}_A \right \}$ basis and $X_2$ is the
  result of measuring $B$ in the $\left \{ \Ket{+}_{B},\Ket{-}_{B}
  \right \}$ basis then, as before, the state $\rho_{B|X_1=x_1}$ would
  have to be a definite state in the $\left \{ \Ket{0}_B, \Ket{1}_B
  \right \}$ basis and $\rho_{B|X_2=x_2}$ would have to be a definite
  state in the $\left \{ \Ket{+}_B,\Ket{-}_B \right \}$ basis.  The
  putative joint state would then have to have a conditional with
  components $\rho_{B|X_1=x_1,X_2=x_2}$ that are definite in both
  bases, which is not possible in the formalism as it currently stands.
\end{example}

\section{Compatibility of quantum states}

\label{Compat}

This section describes our Bayesian approach to the compatibility of
quantum states.  We give alternative derivations of the BFM
compatibility criterion from the point of view of objective and
subjective Bayesianism.  In each case, we begin by reviewing the
corresponding argument in the classical case in order to build
intuition, and draw out the parallels to the quantum case using the
conditional states formalism.

\subsection{Objective Bayesian compatibility}

\label{Compat:Obj}

First consider compatibility for a classical random variable $Y$. For
the objective Bayesian, the only way that two agents' probability
assignments can differ is if they have had access to different data,
so suppose \1 learns the value of a random variable $X_{1}$ and \2
learns the value of a different random variable $X_{2}$. According to
the objective Bayesian, there is a unique prior probability
distribution $P(Y,X_{1},X_{2})$ that both \1 and \2 ought to
initially assign to the three variables before they have observed the
values of the $X_{j}$'s. Both \1 and \2's prior distribution for
$Y$ alone is simply the marginal $P(Y) = \sum_{X_{1},X_{2}}
P(Y,X_{1},X_{2})$. Upon observing a particular value $x_{j}$ of $X_{j}$,
\1 and \2 update to their posterior distributions $P(Y|X_{j} =
x_{j})$.

Now suppose that we don't know the details of how \1 and \2 arrived at
their probability assignments and we are simply told that, at some
specific point in time, \1 assigns some distribution $Q_{1}(Y)$ to $Y$
and \2 assigns a distribution $Q_{2}(Y)$ (different from $Q_{1}(Y)$ in
general). For the objective Bayesian, this can only arise in the
manner described above, so the notion of compatibility is defined as
follows.

\begin{definition}[Classical objective Bayesian compatibility]
  \label{def:Compat:Obj}
  Two probability distributions $Q_{1}(Y)$ and $Q_{2}(Y)$ are
  \emph{compatible} if it is possible to construct a pair of random
  variables, $X_{1}$ and $X_{2}$, and a joint distribution
  $P(Y,X_{1},X_{2})$ such that $Q_{1}(Y)$ can be obtained by Bayesian
  conditioning on $X_{1}=x_{1}$ for some value $x_{1},$ and $Q_{2}(Y)$
  can be obtained by Bayesian conditioning on $X_{2}=x_{2}$ for some
  value $x_{2},$ that is,
  \begin{equation}
    Q_{j}(Y)=P\left(  Y|X_{j}=x_{j}\right)
  \end{equation}
  for some values $x_{j}$ of $X_{j}$.  Further, we require that
  $P(X_1=x_1,X_2=x_2) \neq 0$ so that there is a possibility for both
  outcomes to be obtained simultaneously.
\end{definition}

This definition of compatibility is equivalent to the requirement that
the supports of $Q_1(Y)$ and $Q_2(Y)$ have nontrivial intersection,
where the support of a probability distribution $P(Y)$ is defined as
\textrm{supp}$\left[ P(Y)\right] \equiv\left\{ y\mid P\left(
    Y=y\right) >0\right\}$.

\begin{theorem}
  \label{thm:Compat:Obj}
  Two distributions $Q_{1}(Y)$ and $Q_{2}(Y)$ satisfy
  definition~\ref{def:Compat:Obj}, i.e.\ they are compatible in the
  objective Bayesian sense, iff they share some common support, i.e.
  \begin{equation}
    \emph{supp}\left[  Q_{1}(Y)\right]  \cap \emph{supp} \left[
      Q_{2}(Y)\right]  \neq \emptyset.
  \end{equation}
\end{theorem}

The proof makes use of the following lemma.

\begin{lemma}
  \label{lem:Compat:DistSupp}
  If a probability distribution $P(X,Y)$ satisfies $P(X=x) \neq 0$
  then $\text{supp} \left[ P(Y|X=x)\right] \subseteq \text{supp}
  \left[ P(Y)\right]$.
\end{lemma}

\begin{proof}
  The condition $P(X=x) \neq 0$ implies that $P(Y=y|X=x)$ is well
  defined for all $y$.  Let $\ker \left [ P(Y) \right ] = \left \{ y
    \mid P(Y=y) = 0 \right \}$, i.e. $\ker \left [ P(Y) \right ]$ is
  the complement of $\text{supp} \left [ P(Y) \right ]$.  Let $y \in
  \ker \left [ P(Y) \right ]$.  Since $P(Y=y)=0$, we have $\sum_{x'}
  P(Y=y,X=x')=0$, which implies that $P(Y=y,X=x')=0$ for every value
  $x'$ and consequently that $P(Y=y|X=x)=0$.  In other words,
  $y\in\ker\left[ P\left( Y\right) \right] $ implies $y\in\ker\left[
    P\left( Y|X=x\right) \right] $, which means that $\ker\left[
    P\left( Y\right) \right] \subseteq\ker\left[ P\left(Y|X=x\right)
  \right] $, or equivalently $\text{supp} \left[ P(Y|X=x)\right]
  \subseteq \text{supp}\left[ P(Y)\right]$.
\end{proof}

\begin{proof}[Proof of theorem~\ref{thm:Compat:Obj}]\mbox{\phantom{M}}\vspace{1em} \\
  \textbf{The ``only if'' half}: \\
  It is given that there is a joint distribution $P \left (
    Y,X_{1},X_{2}\right ) $ such that $Q_{j}(Y)=P\left( Y|X_j =x_j
  \right ) $.  Since $P(X_1=x_1,X_2=x_2) \neq 0$,
  $P(Y|X_1=x_1,X_2=x_2)$ exists and, by
  lemma~\ref{lem:Compat:DistSupp}, it must satisfy
  \begin{multline}
    \text{supp}\left[ P\left( Y|X_{1}=x_{1},X_{2}=x_{2}\right) \right]
    \subseteq \\ \text{supp}\left[  P\left(  Y|X_{1}=x_{1}\right)
    \right]
  \end{multline}
  \begin{multline}
    \text{supp}\left[ P\left( Y|X_{1}=x_{1},X_{2}=x_{2}\right) \right]
    \subseteq \\ \text{supp}\left[ P\left( Y|X_{2}=x_{2}\right)
    \right].
  \end{multline}
  Since every probability distribution has nontrivial support,
  $\text{supp} [P\left( Y|X_{1}=x_{1},X_{2}=x_{2}\right)] \neq
  \emptyset$, so
  \begin{equation}
    \text{supp}\left[  P\left(  Y|X_{1}=x_{1}\right)  \right]  \cap\text{supp}
    \left[  P\left(  Y|X_{2}=x_{2}\right)  \right]  \neq \emptyset.
  \end{equation}

  \noindent\textbf{The ``if'' half}:\\
  Given that $Q_{1}(Y)$ and $Q_{2}(Y)$ have intersecting support, one
  can find a normalized probability distribution $Q_0 \left( Y\right)
  $ such that
  \begin{align}
    Q_{1}(Y) & = p_{1} Q_0 \left( Y\right) +\left( 1-p_{1}\right)
    Q'_{1}(Y), \label{eq:Compat:jj1}\\
    Q_{2}(Y) & = p_{2} Q_0 \left( Y\right) +\left( 1-p_{2}\right)
    Q'_{2}(Y), \label{eq:Compat:jj2}
  \end{align}
  where $Q'_{1}(Y)$ and $Q'_{2}(Y)$ are each normalized probability
  distributions and $0<p_{1},p_{2}\leq 1$.

  This decomposition can be used to construct two random variables,
  $X_1$ and $X_2$, and a joint distribution $P\left(
    Y,X_{1},X_{2}\right) $ such that $P(X_1=x_1,X_2=x_2) \neq 0$ and
  $Q_{j}(Y)=P\left( Y|X_{j}=x_{j}\right)$ for some values $x_1$ and
  $x_2$.  Let $X_{1}$ and $X_{2}$ be bit-valued variables that take
  values $\{0,1\}$, and define
  \begin{align}
    P\left(  Y|X_{1}=0,X_{2}=0\right) & = Q_0(Y) \\
    P\left(  Y|X_{1}=0,X_{2}=1\right) & = Q'_1(Y) \\
    P\left(  Y|X_{1}=1,X_{2}=0\right) & = Q'_2(Y).
  \end{align}
  The result of conditioning on $(X_{1}=1,X_{2}=1)$ can be taken to be
  an arbitrary distribution, denoted by $N(Y)$, i.e.
  \begin{equation}
    P\left(  Y|X_{1}=1,X_{2}=1\right)  =N(Y).
  \end{equation}
  Next, define the following distribution over $X_{1}$ and $X_{2}$:
  \begin{align}
    P\left ( X_1 = 0, X_2 = 0 \right ) &  = p_{1}p_{2} \\
    P\left ( X_1 = 0, X_2 = 1 \right ) & = (1-p_{1})p_{2} \\
    P\left ( X_1 = 1, X_2 = 0 \right ) & = p_{1} \left ( 1-p_{2}
    \right ) \\
    P\left ( X_1 = 1, X_2 = 1 \right ) & = \left ( 1-p_{1} \right )
    \left ( 1-p_{2} \right ).
  \end{align}

  Using these, we can define $P(Y,X_1,X_2) = P(Y|X_1,X_2)P(X_1,X_2)$.
  It is straightforward to verify that this satisfies $P(X_1=0,X_2=0)
  = p_1p_2 >0$ and that $P\left( Y|X_{1}=0\right) $ and $P\left(
    Y|X_{2}=0\right) $ are equal to the right-hand sides of
  eqs.~\eqref{eq:Compat:jj1} and \eqref{eq:Compat:jj2}.  Consequently,
  they are equal to $Q_{1}(Y)$ and $Q_{2}(Y)$ respectively.
\end{proof}

The ``only if'' part of the proof of theorem~\ref{thm:Compat:Obj}
establishes that intersecting supports is a necessary requirement for
objective Bayesian state assignments.  On the other hand, the ``if''
part only establishes sufficiency for causal scenarios that support
generic joint states.  For a given causal scenario, i.e.\ a given set
of causal relations holding among $Y$, $X_1$ and $X_2$, there may be
restrictions on the joint probability distributions that can arise.
As an extreme example, if the causal structure is such that the
composite variable $YX_1$ and the elementary variable $X_2$ are
neither connected by some direct or indirect causal influence, nor
connected by a common cause, then they will be statistically
independent and the joint distribution will factorize as
$P(Y,X_1,X_2)=P(Y,X_1)P(X_2)$.  Under such a restriction, there are
certain pairs of states $Q_1(Y)$ and $Q_2(Y)$ that have intersecting
support, but \1 and \2 could never come to assign them by conditioning
on $X_1$ and $X_2$.  For instance, in the example just mentioned,
$Q_2(Y)$ must be equal to the prior over $Y$ and consequently, by
lemma~\ref{lem:Compat:DistSupp}, the only pairs $Q_1(Y)$ and $Q_2(Y)$
that can arise by such conditioning are pairs for which the support of
$Q_1(Y)$ is contained in that of $Q_2(Y)$. Therefore, not every pair
of compatible state assignments will arise in a given causal scenario.
On the other hand, in ``generic'' scenarios wherein the causal
structure does not force any conditional independences in the joint
distribution over $Y$, $X_1$ and $X_2$, the ``if'' part of the proof
does establish that any pair of states with intersecting support can
arise as the state assignments of a pair of objective Bayesian agents.

Turning now to the quantum case, consider a quantum region $B$ with
Hilbert space $\Hilb[B]$. For the objective Bayesian the only way that
two agents' state assignments to $B$ can differ is if they have access
to different data.  We represent this data by two random variables
$X_{1}$ and $X_{2}$, where \1 has access to $X_{1}$ and \2 has access
to $X_{2}$. Assume that the causal scenario of the experiment can be
described by a joint state on the hybrid region $BX_1X_2$, as
discussed in \S\ref{Model:Two}.

Given that this is an objective Bayesian approach, before \1 and \2
observe the values of the $X_{j}$'s, there is a unique prior state
$\rho_{BX_{1}X_{2}}$ which they should both assign, the prior state for $B$
alone being $\rho_{B} = \Tr[X_1X_2]{\rho_{BX_1X_2}}$.  After \1 and \2
observe the values $x_{j}$ for $X_{j}$ they update their states for $B$ to the
posteriors $\rho_{B|X_{j}=x_{j}}$.

Now suppose that we don't know the details of how \1 and \2
arrived at their state assignments and we are simply told that, at
some specific point in time, \1 assigns a state $\sigma_{B}^{(1)}$
to $B$ and \2 assigns a state $\sigma_{B}^{(2)}$ (different from
$\sigma_{B}^{(1)}$ in general). For the objective Bayesian, this can
only arise in the manner described above, so the condition for
compatibility is that it should be possible to construct a hybrid
state $\rho_{BX_{1}X_{2}}$ over $B$ and two classical random variables
$X_{1}$ and $X_{2}$ such that $\sigma_{B}^{(j)}=\rho_{B|X_{j}=x_{j}}$
for some values $x_{j}$ of $X_{j}$.

\begin{definition}[Quantum objective Bayesian compatibility]
  \label{def:Compat:QObj}
  Two quantum states $\sigma_{B}^{(1)}$ and $\sigma_{B}^{(2)}$ of a
  region $B$ are compatible if it is possible to construct a pair of
  random variables $X_{1}$ and $X_{2}$ and a hybrid state
  $\rho_{BX_{1}X_{2}}$ such that $\sigma_{B}^{(1)}$ can be obtained by
  Bayesian conditioning on $X_{1}=x_{1}$ for some value $x_{1},$ and
  $\sigma_{B}^{(2)}$ can be obtained by Bayesian conditioning on
  $X_{2}=x_{2}$ for some value $x_{2}$, i.e.
  \begin{equation}
    \sigma_{B}^{(j)}=\rho_{B|X_{j}=x_{j}}
  \end{equation}
  for some values $x_{j}$ of $X_{j}$.  Further, we require that
  $\rho_{X_1=x_1,X_2=x_2} \neq 0$ so that there is a possibility for
  both outcomes to be obtained simultaneously.
\end{definition}

This holds whenever the BFM compatibility condition is satisfied, as
the following theorem demonstrates.  Recall that the support of a
state $\rho_{B}$ is the span of the eigenvectors of $\rho_{B}$
associated with nonzero eigenvalues.  We denote it by $\text{supp} \left
  [ \rho _{B}\right ]$.

\begin{theorem}
  \label{thm:Compat:QObj}
  Two quantum states $\sigma_{B}^{(1)}$ and $\sigma_{B}^{(2)}$ of a
  region $B$ satisfy definition~\ref{def:Compat:QObj}, i.e.\ they are
  compatible in the objective Bayesian sense, iff they share some
  common support, i.e.
  \begin{equation}
    \emph{supp} \left[  \sigma_{B}^{(1)}\right]  \cap\emph{supp}
  \left[  \sigma_{B}^{(2)}\right]  \neq\emptyset,
  \end{equation}
  where $\cap$ indicates the geometric intersection of the subspaces.
\end{theorem}

The proof of this theorem closely resembles the proof of its classical
counterpart. First, note the quantum analogue of
lemma~\ref{lem:Compat:DistSupp}.

\begin{lemma}
  \label{lem:Compat:QDistSupp}
  If a hybrid state $\rho_{XB}$ satisfies $\rho_{X=x} \neq 0$ then
  $\text{supp}\left[ \rho_{B|X=x}\right] \subseteq \text{supp}\left[
    \rho_{B}\right]$.
\end{lemma}

\begin{proof}
  The condition $\rho_{X=x} \neq 0$ implies that $\rho_{B|X=x}$ is
  well defined.  Let $\ker \left [ \rho_B \right ] = \left \{
    \Ket{\psi}_B \mid \rho_B \Ket{\psi}_B = 0 \right \}$, i.e. $\ker
  \left [ \rho_B \right ]$ is the orthogonal complement of
  $\text{supp} \left [ \rho_B \right ]$.  Let $\Ket{\psi}_B \in \ker
  \left [ \rho_B \right ]$.  Then $\Bra{\psi}_B \Tr[X]{\rho_{BX}}
  \Ket{\psi}_B =0$.  This implies that $\Bra{\psi}_B \rho_{B,X=x'}
  \Ket{\psi}_B=0$ for every $x'$ because each operator $\rho_{B,X=x'}$
  is positive.  Consequently, $\Bra{\psi}_B \rho_{B|X=x} \Ket{\psi}_B
  =0$.  In other words, if $\Ket{\psi}_B \in \ker \left [ \rho_{B}
  \right ]$ then $\Ket{\psi}_B \in \ker \left [ \rho_{B|X=x} \right
  ]$, which means that $\ker\left[ \rho_{B}\right] \subseteq\ker\left[
    \rho_{B|X=x}\right] $, or equivalently $\text{supp} \left[
    \rho_{B|X=x}\right] \subseteq \text{supp}\left[ \rho_{B}\right]$.
\end{proof}

\begin{proof}[Proof of theorem~\ref{thm:Compat:QObj}]
  \mbox{\phantom{M}} \vspace{1em}\\
  \textbf{The ``only if'' half}: \\
  It is given that there is a hybrid joint state $\rho_{BX_{1}X_{2}}$
  such that $\sigma_{B}^{(j)}=\rho_{B|X_{j}=x_{j}}$ for some values
  $x_{j}$ of $X_{j}$.  Since $\rho_{X_1=x_1,X_2=x_2} \neq 0$, the
  conditional state $\rho_{B|X_{1}=x_{1},X_{2}=x_{2}}$ is well
  defined. Lemma~\ref{lem:Compat:QDistSupp} then implies that
  \begin{align}
    \text{supp}\left[ \rho_{B|X_{1}=x_{1},X_{2}=x_{2}}\right] &
    \subseteq
    \text{supp}\left[  \rho_{B|X_{1}=x_{1}}\right]  \\
    \text{supp}\left[ \rho_{B|X_{1}=x_{1},X_{2}=x_{2}}\right] &
    \subseteq \text{supp}\left[ \rho_{B|X_{2}=x_{2}}\right].
  \end{align}
  Since $\rho_{B|X_{1}=x_{1},X_{2}=x_{2}}$ has nontrivial
  support, it follows that
  \begin{equation}
    \text{supp}\left[  \rho_{B|X_{1}=x_{1}}\right]  \cap\text{supp}\left[
      \rho_{B|X_{2}=x_{2}}\right]  \neq \emptyset.
  \end{equation}

  \noindent\textbf{The ``if'' half}: \\
  Given that $\sigma _{B}^{(1)}$ and $\sigma_{B}^{(2)}$ have
  intersecting support, one can find a quantum state $\mu_{B}$
  such that
  \begin{align}
    \sigma_{B}^{(1)} & =p_{1}\mu_{B} + \left( 1-p_{1}\right)
    \eta_{B}^{(1)}, \label{eq:Compat:qq1} \\
    \sigma_{B}^{(2)} & =p_{2}\mu_{B}+\left( 1-p_{2}\right)
    \eta_{B}^{(2)}, \label{eq:Compat:qq2}
  \end{align}
  where $\eta_{B}^{(1)}$ and $\eta_{B}^{(2)}$ are each quantum states
  and $0<p_{1},p_{2}\leq 1$.

  This decomposition can be used to construct two classical variables,
  $X_1$ and $X_2$ and a hybrid state $\rho_{BX_{1}X_{2}}$ such that
  $\rho_{X_1 = x_1,X_2 = x_2} \neq 0$ and $\sigma
  _{B}^{(j)}=\rho_{B|X_{j}=x_{j}}$ for some values $x_1$ and $x_2$. \
  Let $X_{1}$ and $X_{2}$ be bit-valued variables, and define
  \begin{align}
    \rho_{B|X_{1}=0,X_{2}=0} & = \mu_{B} \\
    \rho_{B|X_{1}=0,X_{2}=1} & =\eta_{B}^{(1)} \\
    \rho_{B|X_{1}=1,X_{2}=0} & =\eta_{B}^{(2)}.
  \end{align}
  The result of conditioning on $(X_{1}=1,X_{2}=1)$ can be taken to be
  an arbitrary state, denoted $\nu_{B}$, i.e.
  \begin{equation}
    \rho_{B|X_{1}=1,X_{2}=1}=\nu_{B}.
  \end{equation}
  Next, define the following (classical) state over $X_{1}$ and
  $X_{2}$:
  \begin{multline}
    \rho_{X_{1},X_{2}} =p_{1}p_{2}\Ket{00}\Bra{00}_{X_1X_2}
    \\ +(1-p_{1})p_{2} \Ket{01}\Bra{01}_{X_1X_2} \\
    +p_{1}\left( 1-p_{2}\right) \Ket{10}\Bra{10}_{X_1X_2} \\ + \left(
      1-p_{2}\right) \left( 1-p_{2}\right) \Ket{11}\Bra{11}_{X_1X_2}.
  \end{multline}
  This can be combined with the conditional states defined above to obtain
  \begin{multline}
    \rho_{BX_{1}X_{2}} = \rho_{B|X_1 X_2} \rho_{X_1 X_2} \\
    = p_{1}p_{2}\left( \mu_{B}\otimes
      \Ket{00}\Bra{00}_{X_1X_2} \right) \\
    +(1-p_{1})p_{2}\left( \eta_{B}^{(1)}\otimes
      \Ket{01}\Bra{01}_{X_1X_2} \right) \\
    +p_{1}\left( 1-p_{2}\right) \left(
      \eta_{B}^{(2)}\otimes \Ket{10}\Bra{10}_{X_1X_2} \right) \\
    + \left( 1-p_{2}\right) \left( 1-p_{2}\right) \left( \nu_{B}
      \otimes \Ket{11}\Bra{11}_{X_1X_2} \right).
  \end{multline}
  It is then easy to verify that $\rho_{B|X_{1}=0}$ and
  $\rho_{B|X_{2}=0}$ are equal to the right-hand sides of
  eqs.~\eqref{eq:Compat:qq1} and \eqref{eq:Compat:qq2}, and
  consequently are equal to $\sigma_{B}^{(1)}$ and $\sigma_{B}^{(2)}$
  respectively.
\end{proof}

As noted in the classical case, the ``if'' part of the proof only
establishes sufficiency of the BFM criterion for causal scenarios that
support generic joint states.  Certain causal scenarios may enforce a
restriction on the pairs of states $\sigma^{(1)}_B$ and
$\sigma^{(2)}_B$ that \1 and \2 can come to assign by conditioning on
$X_1$ and $X_2$.  For instance, consider the causal scenarios depicted
in fig.~\ref{fig:Model:Suff}, where $X_2$ is obtained by
post-processing of $X_1$, so that all correlations between $X_2$ and
$B$ are mediated by $X_1$.  In this case, the only pairs
$\sigma^{(1)}_B$ and $\sigma^{(2)}_B$ that can arise by conditioning
on $X_1$ and $X_2$ are those for which the support of $\sigma^{(1)}_B$
is contained in that of $\sigma^{(2)}_B$.  We are led to the same
conclusion as we found classically: although BFM compatibility is
necessary in any causal scenario, not every pair of BFM compatible
state assignments can arise in every causal scenario.  Nonetheless, we
can always find a causal scenario wherein there are no restrictions on
the joint state $\rho_{BX_1X_2}$ and therefore no restriction on the
states to which a pair of agents can be led by Bayesian conditioning.
The causal scenario considered by BFM, where $X_1$ and $X_2$ are the
outcomes of a pair of remote measurements on $B$ (depicted in
fig.~\ref{fig:Model:TwoRemote}) is one such example, as is the causal
scenario considered by Jacobs, where $X_1$ and $X_2$ are the outcomes
of a sequential pair of measurements and $B$ is the output (depicted
in fig.~\ref{fig:Model:TwoRemote})\footnote{To realize an arbitrary
  joint state $\rho_{BX_1X_2}$ in this scenario, it suffices to use
  the following components.  Let $A_1$ and $A_2$ be classical systems
  that each have a preferred basis that is labeled by the values of
  both $X_1$ and $X_2$. Set the state of the input $A_1$ to
  $\rho_{A_1} = \sum_{x_1,x_2} P(X_1=x_1,X_2=x_2)
  \Ket{x_1x_2}\Bra{x_1x_2}_{A_1}$, where $P(X_1=x_1,X_2=x_2)= \rho_{
    X_1=x_1,X_2=x_2}$.  Let the first quantum instrument,
  $\varrho_{X_1A_2|A_1}$, be a measurement of the projector-valued
  measure $\{ \Ket{x_1}\Bra{x_1}\otimes I \}$ with a projection
  postulate update rule, and let the second quantum instrument,
  $\varrho_{X_2 B|A_2}$, be a measurement of the projector-valued
  measure $\{ I \otimes \Ket{x_2}\Bra{x_2} \}$ with an update rule
  that takes $\Ket{x_1x_2}\Bra{x_1x_2}_{A_2}$ to
  $\rho_{B|X_1=x_1,X_2=x_2}$.}.

\subsection{Subjective Bayesian compatibility}

A subjective Bayesian cannot use the approach just discussed in
general, since it depends on postulating a unique prior state over
$B$, $X_{1}$, and $X_{2}$ (or $Y$, $X_{1}$, and $X_{2}$ in the
classical case) that all agents agree upon before collecting their
data. Given that the choice of prior is an unanalyzable matter of
belief for the subjectivist, there is no reason why \1 and \2 need to
agree on a prior at the outset and, further, there is no reason why
the difference in their probability assignments has to be explained by
their having had access to different data in the first place. If it
happens that \1 and \2 did have a shared prior before collecting their
data then the argument runs through, but for the subjective Bayesian
this is the exception rather than the rule. In fact, since subjective
Bayesians do not rule out as irrational the possibility of agents
starting out with contradictory beliefs, it might seem that there is
no role for compatibility criteria in this approach at all.

However, this is not the case since, although subjective Bayesians do
not analyze how agents arrive at their beliefs, they are interested in
whether it is possible for them to reach inter-subjective agreement in
the future, i.e.\ whether it is possible for them to resolve their
differences by experiment or whether their disagreement is so extreme
that one of them has to make a wholesale revision of their beliefs in
order to reach agreement. In the classical case, a subjective Bayesian
will therefore say that two probability assignments $Q_{1}(Y)$ and
$Q_{2}(Y)$ to a random variable $Y$ are compatible if it is possible
to construct an experiment, for which \1 and \2 agree upon a
statistical model, i.e.\ a likelihood function $P(X|Y)$, such that at
least one outcome $X=x$ of the experiment would cause \1 and \2 to
assign identical probabilities when they update their probabilities by
Bayesian conditioning.  In other words, the subjective Bayesian
account of compatibility is in terms of the possibility of
\emph{future} agreement, in contrast to the objective Bayesian
account, which relies on a guarantee of agreement in the
\emph{past}\footnote{Note that by ``the possibility of future
  agreement'' we mean agreement about what value the variable $Y$ took
  at some particular moment in its dynamical history, but where that
  agreement might only be achieved in the future epistemological lives
  of the agents, after they have acquired more information.  We do
  \emph{not} mean the possibility of agreement about what value the
  variable $Y$ takes at some future moment in its dynamical history.
  The distinction between these two sorts of temporal relation,
  i.e.\ between moments in the epistemological history of an agent on
  the one hand and between moments in the ontological history of the
  system on the other, is critical when discussing Bayesian inference
  for systems that persist in time.  It is discussed in
  \cite{Leifer2011a} and in \S~\ref{CFS}.}

\begin{definition}[Classical subjective Bayesian compatibility]
  \label{def:Compat:Subj}
  Two probability distributions, $Q_{1}(Y)$ and $Q_{2}(Y),$ are
  compatible if it is possible to construct a random variable $X$ and
  a conditional probability distribution $P(X|Y)$ (often called a
  likelihood function in this context) such that there exists a value
  $x$ of $X$ for which $\sum_Y P(X=x|Y)Q_j(Y) \neq 0$ and
  \begin{equation}
    P_{1}(Y|X=x)=P_{2}(Y|X=x)
  \end{equation}
  where $P_{j}(Y|X)\equiv P\left( X|Y\right) Q_{j}(Y)/\sum_{Y}P\left(
    X|Y\right) Q_{j}(Y)$.
\end{definition}

It turns out that the
mathematical criteria that $Q_1(Y)$ and $Q_2(Y)$ must satisfy in order
to be compatible in this subjective Bayesian sense are precisely the
same as those required for objective Bayesian compatibility.

\begin{theorem}
  \label{thm:Comp:Subj}
  Two probability distributions $Q_{1}(Y)$ and $Q_{2}(Y)$ satisfy
  definition~\ref{def:Compat:Subj}, i.e.\ they are compatible in the
  subjective Bayesian sense, iff they share some common support, i.e.
  \begin{equation}
    \emph{supp}\left[  Q_{1}(Y)\right]  \cap \emph{supp}\left[
      Q_{2}(Y)\right]  \neq \emptyset.
  \end{equation}
\end{theorem}

\begin{proof}\mbox{\phantom{M}}\vspace{1em}\\
  \textbf{The ``only if'' half}: \\
  Since $P_{j}(Y|X=x)$ is derived from $Q_{j}(Y)$ by Bayesian
  conditioning, it follows from lemma~\ref{lem:Compat:DistSupp} that
  \begin{align}
    \text{supp}\left[ P_{1}\left( Y|X=x\right) \right] & \subseteq
    \text{supp}\left[  Q_{1}\left(  Y\right)  \right]  \label{eq:Compat:Css1} \\
    \text{supp}\left[ P_{2}\left( Y|X=x\right) \right] & \subseteq
    \text{supp}\left[ Q_{2}\left( Y\right) \right].\label{eq:Compat:Css2}
  \end{align}
  However, by assumption, $P_{1}\left( Y|X=x\right) =P_{2}\left(
    Y|X=x\right)$, so the left-hand sides of
  eqs.~\eqref{eq:Compat:Css1} and \eqref{eq:Compat:Css2} are equal.
  It follows that $Q_{1}\left( Y\right) $ and $Q_{2}(Y)$ have some
  common support, namely, $\text{supp}\left[ P_{1}\left( S|X=x\right)
  \right]$.

  \noindent\textbf{The ``if'' half}: \\
  By assumption, there is at least one value $y$ of $Y$ belonging to
  the common support of $Q_{1}\left( Y\right) $ and $Q_{2}\left(
    Y\right)$.  Let $X$ be a classical bit and define the likelihood
  function
  \begin{align}
    P(X=0|Y=y) & = 1 & P(X=1|Y=y) & = 0 \\
    P(X=0|Y\neq y) & = 0 & P(X=1|Y\neq y) & = 1.
  \end{align}
  If \1 and \2 agree to use this likelihood function, then, upon
  observing $X=0$, they will update their distributions to
  \begin{multline}
    P_{j}(Y=y'|X=0)= \\
    \frac{P(X=0|Y=y')Q_{j}(Y=y')}{\sum_{y'}P(X=0|Y=y')Q_{j}(Y=y')} \\ =
    \delta_{y,y'},
  \end{multline}
  which is independent of $j$ and hence brings them into agreement.
\end{proof}

In the quantum case, if \1 and \2 assign states $\sigma_{S}^{(j)}$ to
a quantum region then they are compatible if there is some classical
data $X$ that they can collect about the system, for which \1 and
\2 agree upon a statistical model, such that observing at least one
value $x$ of $X$ would cause their state assignments to become
identical.

\begin{definition}[Quantum subjective Bayesian compatibility]
  \label{def:Compat:QSubj}
  Two states $\sigma_{B}^{(1)}$ and $\sigma_{B}^{(2)}$ are compatible
  if it is possible to construct a random variable $X$ and a
  conditional state $\rho_{X|B}$ (which we call a likelihood operator)
  such that there exists a value $x$ of $X$ for which
  $\Tr[B]{\rho_{X=x|B}\sigma_B^{(j)}} \neq 0$ and
  \begin{equation}
    \rho_{B|X=x}^{(1)}=\rho_{B|X=x}^{(2)}
  \end{equation}
  where $\rho_{B|X=x}^{(j)}$ is given by the quantum Bayes' theorem:
  $\rho_{B|X=x}^{(j)}\equiv\left( \rho_{X=x|B}\Sprod
    \sigma_{B}^{(j)}\right) /\Tr[B]{\rho_{X=x|B}\sigma_{B}^{(j)}}$.
\end{definition}

Once again, subjective Bayesian compatibility has the same
mathematical consequences as its objective counterpart.  Both are
equivalent to requiring the BFM criterion.

\begin{theorem}
  \label{thm:Compat:QSubj}
  Two states $\sigma_{B}^{(1)}$ and $\sigma_{B}^{(2)}$ satisfy
  definition~\ref{def:Compat:QSubj}, i.e.\ they are compatible in the
  subjective Bayesian sense, iff they share common support, i.e.
  \begin{equation}
    \emph{supp} \left[  \sigma_{B}^{(1)}\right]
    \cap \emph{supp} \left[  \sigma_{B}^{(2)}\right]  \neq
    \emptyset,
  \end{equation}
  where $\cap$ denotes the geometric intersection.
\end{theorem}

\begin{proof}\mbox{\phantom{M}}\vspace{1em}\\
  \textbf{The ``only if'' half}: \\
  Since $\rho_{B|X=x}^{(j)}$ is derived from $\sigma_{B}^{(j)}$ by
  Bayesian conditioning, it follows from
  lemma~\ref{lem:Compat:QDistSupp} that
  \begin{align}
    \text{supp}\left[ \rho_{B|X=x}^{(1)}\right] &
    \subseteq\text{supp}\left[
      \sigma_{B}^{(1)}\right]  \label{eq:Compat:Qss1} \\
    \text{supp}\left[ \rho_{B|X=x}^{(2)}\right] &
    \subseteq\text{supp}\left[
      \sigma_{B}^{(2)}\right]. \label{eq:Compat:Qss2}
  \end{align}
  However, by assumption, $\rho _{B|X=x}^{(1)}=\rho_{B|X=x}^{(2)}$, so
  the left-hand sides of eqs.~\eqref{eq:Compat:Qss1} and
  \eqref{eq:Compat:Qss2} are equal.  It follows that
  $\sigma_{B}^{(1)}$ and $\sigma_{B}^{(2)}$ have some common support,
  namely, $\mathrm{supp}\left[ \rho_{B|X=x}^{(1)}\right]$.

  \noindent\textbf{The ``if'' half}: \\
  By assumption, the supports of $\sigma_{B}^{(1)}$ and
  $\sigma_{B}^{(2)}$ have nontrivial intersection. It follows that
  there is a pure state $\Ket{\psi}_B \in \Hilb[B]$ in the common
  support.  Let $X$ be a classical bit and define the likelihood
  operator
  \begin{multline}
    \rho_{X|B}= \\ \Ket{0}\Bra{0}_X \otimes \Ket{\psi}\Bra{\psi}_B +
    \Ket{1}\Bra{1}_X \otimes \left(
      I_{B}-\Ket{\psi}\Bra{\psi}_B \right).
  \end{multline}
  If \1 and \2 agree to use this likelihood operator, then, upon
  observing $X=0$ they will update their states to
  \begin{equation}
    \rho_{B|X=0}^{(j)}=\frac{\rho_{X=0|B}\Sprod \sigma_{B}^{(j)}}
    {\Tr[B]{\rho_{X=0|B} \sigma_{B}^{(j)}}}
    =\Ket{\psi}\Bra{\psi}_B,
  \end{equation}
  which is independent of $j$ and hence brings them into agreement.
\end{proof}

\subsection{Comparison to other approaches}

\subsubsection{Brun, Finkelstein and Mermin}

The original BFM argument \cite{Brun2002}, which is objective Bayesian
in flavor, is divided into arguments for the necessity and
sufficiency of their criterion.  To establish sufficiency, they show
that for any pair of state assignments that satisfy their criterion,
one can find a triple of distinct systems, and a quantum state thereon,
such that if \1 measures one system and \2 another, then for some pair
of outcomes \1 and \2 are led to update their description of the third
system to the given pair of state assignments.  This is equivalent to
the ``if'' part of our theorem~\ref{thm:Compat:QObj} when applied to
the remote measurement scenario depicted in
fig.~\ref{fig:Model:TwoRemote}.  The argument provided by BFM for the
necessity of their criterion is based on a set of reasonable-sounding
requirements.  For example, their first requirement is:
\begin{quote}
  If anybody describes a system with a density matrix $\rho$, then
  nobody can find it to be in a pure state in the null space of
  $\rho$.  For although anyone can get a measurement outcome that
  everyone has assigned nonzero probabilities, nobody can get an
  outcome that anybody knows to be impossible.
\end{quote}

If one is adopting an approach wherein quantum states describe the
information, knowledge, or beliefs of agents, then the notion of
finding a system ``to be in a pure state'' is inappropriate, as
emphasized by Caves, Fuchs and Schack \cite{Caves2002}.  However, even
glossing over this, their argument does not satisfy an ideal to which
a proper objective Bayesian account of compatibility should strive,
namely, of being justified by a general methodology for Bayesian
inference.  This ideal is illustrated by the derivation of the
objective Bayesian criterion of \emph{classical} compatibility
presented in \S\ref{Compat:Obj}: if a pair of agents obey the
strictures of objective Bayesianism, i.e.\ assigning the same
ignorance priors and updating their probabilities via Bayesian
conditioning, then they will never encounter a situation in which the
compatibility condition does not hold, and conversely if the
compatibility condition holds, it is always possible for them to come
to their state assignments by Bayesian updating.

Because we have proposed a methodology for \emph{quantum} Bayesian
inference, we can achieve this ideal in the quantum case as well.
Indeed, the close parallel between the proofs of the classical and
quantum compatibility theorems demonstrates that one can achieve the
ideal in the quantum context to precisely the same extent that it can
be achieved in the classical context.  Whilst our argument for
sufficiency of the BFM criterion (the second part of our proof of
theorem~\ref{thm:Compat:QObj}) is mathematically similar to BFM's
argument for sufficiency, it is only against the background of our
framework of quantum conditional states that it becomes possible to
identify the update rule used by \1 and \2 as an instance of
Bayesian conditioning, and thus a special case of a general
methodology for Bayesian inference.

A second point to note is that in our argument for the compatibility
condition, we consider a triple of space-time regions that do not
necessarily correspond to three distinct systems at a given time ---
the case considered by BFM.  The causal relation between them might
instead be any of those depicted in
figs.~\ref{fig:Model:Suff}--\ref{fig:Model:Jacobs}, or indeed any
scenario wherein all the available information about the quantum
region can be captured by assigning a single quantum state.  Thus, our
results generalize the range of applicability of the BFM compatibility
criterion to a broader set of causal scenarios.

\subsubsection{Jacobs}

\label{Compat:Jacobs}

Jacobs \cite{Jacobs2002} has also considered the compatibility of
state assignments using an approach that is objective Bayesian in
flavor.  In his analysis, the region of interest is the output of a
sequence of measurements made one after the other on the same system,
and \1 and \2 have information about the outcomes of distinct subsets
of those measurements.  A simple version of this scenario is where
there is a sequence of \emph{two} measurements, where the outcome of
the first measurement is known to \1 and the outcome of the second is
known to \2.  This is just the causal scenario depicted in
fig.~\ref{fig:Model:Jacobs}, and as emphasized there, such a scenario
falls within the scope of our approach.  In the objective Bayesian
framework, \1 and \2 agree on the input state to the pair of
measurements and they agree on the quantum instruments that describe
each measurement.  Jacobs shows that if \1 and \2's state
assignments are obtained in this way, then they must satisfy the BFM
compatibility criterion, that is, he provides an argument for the
necessity of the BFM compatibility criterion in this causal context.

\newcounter{fn}
\setcounter{fn}{\value{footnote}}
\addtocounter{fn}{-1}

If \1 and \2 come to their state assignments for $B$ using Jacobs'
scheme, then, as explained in \S\ref{Model:Two}, their prior knowledge
of $B$ and the two outcome variables, $X_1$ and $X_2$, can be
described by a joint state $\rho_{BX_1X_2}$.  After observing values
$x_1$ and $x_2$ respectively, they come to assign states
$\rho_{B|X_1=x_1}$ and $\rho_{B|X_2=x_2}$, which are derived from the
conditional states of the joint state $\rho_{BX_1X_2}$.  Such a pair
of states satisfies the definition~\ref{def:Compat:QObj} of quantum
objective Bayesian compatibility.  Theorem~\ref{thm:Compat:QObj} then
implies that their state assignments satisfy the BFM compatibility
criterion.  Conversely, because the set of joint states
$\rho_{BX_1X_2}$ which can arise in this causal scenario is
unrestricted (see footnote \footnotemark[\value{fn}]), it also follows
from theorem~\ref{thm:Compat:QObj} that for any pair of state
assignments satisfying the BFM criterion, \1 and \2 could come to
assign those states in this causal scenario.  These results can be
generalized to the case of a longer sequence of measurements with the
outcomes distributed arbitrarily among a number of parties, which
covers the most general case considered by Jacobs.

To summarize, our results can be applied to Jacobs' scenario and we
recover Jacobs' result that the BFM criterion is a necessary
requirement.  Furthermore, we have improved upon Jacobs' analysis in
two ways.  Firstly, we have shown that the BFM compatibility criterion
is not only a necessary condition for compatibility in this scenario,
but is \emph{sufficient} as well.  Secondly, our analysis demonstrates
that, just as with the scenario of remote measurements, the BFM
criterion can be justified in the scenario of sequential measurements by insisting that
states should be updated by Bayesian conditioning within a general
framework for quantum Bayesian inference.

\subsubsection{Caves, Fuchs and Schack}
\label{CFS}

In contrast to BFM and Jacobs, CFS \cite{Caves2002} discuss the
problem of quantum state compatibility from an explicitly subjective
Bayesian point of view.  They argue that there cannot be a unilateral
requirement to impose compatibility criteria of any sort on subjective
Bayesian degrees of belief because there is no unique prior quantum
state that an agent ought to assign in light of a given collection of
data.  The only necessary constraint is that states should satisfy the
axioms of quantum theory, i.e.\ they should be normalized density
operators.  In particular, it should not be viewed as irrational for
two agents to assign distinct, or even orthogonal, pure states to a
quantum system.

Whilst we agree with this argument, we think that there is still a
role for compatibility criteria within the subjective approach.  They
can be viewed as a check to see whether it is worthwhile for the
agents to engage in a particular inference procedure, and this is
conceptually distinct from viewing them as unilateral requirements that
must be imposed upon all state assignments.  In the case of BFM
compatibility, the criterion of intersecting supports is simply a
check that agents can apply to see if it is worth their while to try
and resolve their differences empirically by collecting more data, or
whether their disagreement is so extreme that resolving it requires
one or more of the agents to make a wholesale revision of their
beliefs.  From this point of view, BFM plays the same role as the
criterion of overlapping supports does in classical subjective
Bayesian probability.

Despite their skepticism of compatibility criteria, CFS do attempt to
recast the necessity part of the BFM argument in terms that would be
more acceptable to the subjective Bayesian, i.e.\ they outline a
series of requirements that a pair of subjective Bayesian agents may
wish to adopt that would lead them to assign BFM compatible states.
They do not provide an argument for sufficiency, so this is one way in
which our argument is more complete.  CFS's argument is quite similar
to the BFM necessity argument, except that it is phrased in terms that
would be more acceptable to a subjective Bayesian.  For example, they
talk about the ``firm beliefs'' of agents rather than saying that
systems are ``found to be'' in certain pure states.  However, this
line of argument is still open to an objection that we leveled against the
BFM argument.  In our view, compatibility criteria should be derived
from the inference methodologies that are being used by the agents
rather than from a list of reasonable sounding requirements.  Another
objection is that their argument relies on strong Dutch Book
coherence, which is a strengthening of the usual Dutch Book coherence
that subjective Bayesians use to derive the structure of classical
probability theory.  Strong coherence entails that if an agent assigns
probability one to an event then she must be certain that it will
occur.  This is obviously problematic in the case of infinite sample
spaces due to the presence of sets of measure zero and, since there is
nothing in the Dutch Book argument that singles out finite sample
spaces, it would not usually be accepted by subjective Bayesians in
that case either.

Since CFS do not believe that the BFM criterion is a uniquely
compelling requirement, they also introduce a number of weaker
compatibility criteria based on the compatibility of the probability
distributions obtained by making different types of measurement on the
system.  Three of these compatibility criteria are equivalent to the
usual intersecting support criterion in the classical case, but they
become inequivalent when applied to quantum theory.  Presumably, this
is supposed to cast doubt upon the uniqueness of BFM as a compelling
compatibility criterion in the quantum case.  However, in our view,
the non-BFM criteria in the CFS hierarchy are not meaningful as
compatibility criteria.  To explain why, we take their weakest
criterion --- $W'$ compatibility --- as an example.

The $W'$ criterion says that two quantum states are compatible if
there exists a measurement for which the Born rule outcome probability
distributions computed from the two states are compatible in the
classical sense, i.e.\ they have intersecting support in the set of
outcomes.  It is fairly easy to see that this places no constraint at
all on state assignments --- such a measurement can always be found.
For example, if \1 and \2 assign two orthogonal pure states then a
measurement in a complementary basis would always yield compatible
probability distributions over the set of outcomes.  CFS argue that \1
and \2 could resolve their differences empirically by making such a
measurement in this scenario.  After the measurement, if both \1 and
\2 learn the outcome and apply the projection postulate, then they
would end up assigning the same state to the system, specifically, the
state in the complementary basis corresponding to the outcome that was
observed.

However, in our view, this does not resolve the original conflict
between \1 and \2.  Although \1 and \2's state assignments to the
region \emph{after} the measurement (its quantum output) are now
identical, their state assignments to the region \emph{before} the
measurement (its quantum input) remain unchanged.  As stated in
\S\ref{Model:Single}, and explained in more detail in
\cite{Leifer2011a}, the state of the region before the measurement
updates via quantum Bayesian conditioning rather than by the
projection postulate. Pure states are fixed points of quantum Bayesian
conditioning, so \1 and \2 will always continue to disagree about the
state of this region, whatever information they later acquire about
the region.

The mistake that CFS have made is to think of compatibility in terms
of persistent systems rather than spatio-temporal regions, and to
think of the projection postulate as a quantum analogue of Bayesian
conditioning.  It is easy to make this mistake because in a classical
theory of Bayesian inference, a measurement can be non-disturbing.  In
this case, the value of the variable $Y$ being measured is not changed
by the measurement, and the update rule for the probability
distribution of $Y$ can be understood as conditioning $Y$ on the
outcome of the measurement.  The variable describing the system before
the measurement is the same as the one describing it after, so that
updating your beliefs about one is the same as updating your beliefs
about the other.  But this is no longer true for classical
measurements that \emph{disturb} the system, and as argued in
\cite{Leifer2011a}, all nontrivial quantum measurements are analogous
to these.  Therefore, to highlight the problem with the $W'$
compatibility criterion, we consider what it would predict in the case
of a disturbing classical measurement.

Suppose the system is a coin that has just been flipped, but is
currently hidden from \1 and \2.  If \1 believes that the coin has
definitely landed heads and \2 believes that it has definitely landed
tails, then their beliefs are certainly incompatible.  If the coin is
then flipped again and \1 and \2 are shown the outcome of the second
toss, they will agree on the current state of the coin, and hence
their state assignments to the system after the observation are now
compatible.  However, because the configuration of the coin was
disturbed in the process of measurement, there is no sense in which
their disagreement about the outcome of the first coin flip has been
resolved.  Similarly, we believe that because nontrivial quantum
measurements always entail a disturbance (in the sense described in
\cite{Leifer2011a}), coming to agreement about the state of the region
after the measurement does not resolve a disagreement about the state
of the region before the measurement.
 
Despite our reservations about the CFS compatibility criteria, they
are still of some independent interest.  In particular, one of them
(the PP criterion) was recently used in a quite different context as
part of a no-go theorem for certain types of hidden variable models
for quantum theory \cite{Barrett2011}.

\section{Intermezzo: conditional independence and sufficiency}

Having dealt with state compatibility, our next task is to develop a
Bayesian approach to combining state assignments.  In order to do
this, two additional concepts are needed: conditional independence and
sufficient statistics, which are reviewed in this section.  Quantum
conditional independence has been studied in \cite{Leifer2008}, from
which we quote results without proof.

\subsection{Conditional independence}

\label{CI}

\subsubsection{Classical conditional independence}

A pair of random variables $R$ and $S$ are conditionally independent
given another random variable $T$ if they satisfy any of the following
equivalent conditions:

\begin{ccindep}

\item \label{item:CI:Def1} $P(S|R,T)=P(S|T)$

\item \label{item:CI:Def2} $P(R|S,T)=P(R|T)$

\item \label{item:CI:Def3} $I(R:S|T)=0$

\item \label{item:CI:Def4} $P(R,S|T)=P(R|T)P(S|T)$,

\end{ccindep}
where it is left implicit that these equations only have to hold for
those values of the variables for which the conditionals are well-defined.
Here, $I(R:S|T)$ is the conditional mutual information of $R$ and $S$
given $T$, defined by
\begin{multline}
  I(R:S|T) = H(R,T)+H(S,T) \\ -H(T)-H(R,S,T),
\end{multline}
where $H(R) = - \sum_r P(R=r)\log_2 P(R=r)$ is the Shannon entropy of
$R$, with the obvious generalization to multiple variables.  Note that
the conditional mutual information is always positive.

Conditional independence of $R$ and $S$ given $T$ means that any correlations between $R$ and
$S$ are mediated, or screened-off, by $T$.  In other words, if one
were to learn the value of $T$ then $R$ and $S$ would become
independent.

\subsubsection{Quantum conditional independence for acausally related
  regions}

In the quantum case, the three random variables $R$, $S$ and $T$
become quantum regions with Hilbert spaces $\Hilb[A]$, $\Hilb[B]$ and
$\Hilb[C]$.  We specialize to the case of three acausally related
regions because the theory of conditional independence has not yet
been developed for other causal scenarios.  Prior to the
introduction of conditional states, it was not obvious whether the
conditional independence conditions \ref{item:CI:Def1},
\ref{item:CI:Def2} and \ref{item:CI:Def4} had quantum analogs, but
\ref{item:CI:Def3} has a straightforward generalization where
$I(A:B|C)$ is now the quantum conditional mutual information defined
as
\begin{multline}
I(A :B|C) =S(A,C)+S(B,C) \\ -S(C)-S(A,B,C),
\end{multline}
where $S(A)=-\Tr[A]{\rho_A \log \rho_A}$ is the von Neumann entropy of
the state on $A$.  The quantum conditional mutual information
satisfies $I(A:B|C) \geq 0$, which is equivalent to the strong
sub-additivity inequality \cite{Lieb1973}, and so the quantum
conditional independence condition $I(A:B|C) = 0$ is the equality
condition for strong sub-additivity.

In the conditional states formalism, there are direct analogs of the
conditions \ref{item:CI:Def1} and \ref{item:CI:Def2} that provide an
alternative characterization of quantum conditional independence.
\begin{theorem}
  \label{thm:CI:QCI}
  For three acausally related regions, $A$, $B$ and $C$, the following
  conditions are equivalent:
  \begin{qcindep}
  \item \label{item:CI:QDef1}$\rho_{A|BC} = \rho_{A|C}$

  \item \label{item:CI:QDef2}$\rho_{B|AC} = \rho_{B|C}$

  \item \label{item:CI:QDef3}$I(A:B|C) = 0$
  \end{qcindep}
\end{theorem}

Due to these equivalences, any of
\ref{item:CI:QDef1}--\ref{item:CI:QDef3} can be viewed as the
definition of quantum conditional independence.

It is also true that
\begin{theorem}
  \label{thm:CI:ComProd}
  If $A$ is conditionally independent of $B$ given $C$, then
  \begin{qcindep}[start=4]
  \item \label{item:CI:QDef4}$\rho_{AB|C} = \rho_{A|C} \rho_{B|C}$.
  \end{qcindep}
\end{theorem}

Because $\rho_{AB|C}$ is self-adjoint, theorem~\ref{thm:CI:ComProd}
implies that $\rho_{A|C}$ and $\rho_{B|C}$ must commute when $A$ and
$B$ are conditionally independent given $C$.  Unlike in the classical
case, the converse of theorem~\ref{thm:CI:ComProd} does not hold,
i.e. $\rho_{AB|C} = \rho_{A|C} \rho_{B|C}$ does not imply conditional
independence.  Extra constraints on the form of $\rho_C$ can be
imposed to yield equivalence, but these are not important for present
purposes (see \cite{Leifer2008} for details).

\subsubsection{Hybrid conditional independence}

The case that is most relevant to the present work is when two
classical random variables $X_1$ and $X_2$ are conditionally
independent given a quantum region $B$.  The proofs of
theorems~\ref{thm:CI:QCI} and \ref{thm:CI:ComProd} only depend on the
existence of a joint state (positive, normalized, density operator)
for the three regions under consideration.  Therefore, if we
specialize to causal scenarios in which a joint state $\rho_{BX_1X_2}$
can be assigned, as discussed in \S\ref{Model:Two}, then the
definitions \ref{item:CI:QDef1}--\ref{item:CI:QDef3} can now be
applied in any of these causal scenarios by substituting $X_1$ for
$A$, $X_2$ for $B$ and $B$ for $C$.  The consequence
\ref{item:CI:QDef4} also applies to this case.

\subsection{Sufficient statistics}

\label{Suff}

The idea of a sufficient statistic can be motivated by a typical
example problem in statistics: estimating the bias of a coin from a
sequence of coin flips that are judged to be independent and
identically distributed.  In this problem, only the relative frequency
of occurrence of heads and tails in the sequence is relevant to the
bias, whilst the exact ordering of heads and tails is irrelevant.  The
relative frequency is then an example of a sufficient statistic for
the sequence with respect to the bias.  In this section, this notion is generalized to the
hybrid case wherein the classical parameter to be estimated is
replaced by a quantum region, but the data is still classical, i.e.\
this section concerns sufficient statistics for classical data with
respect to a quantum region.  Note that quantum sufficient statistics
have been considered before in the literature \cite{Jencova2006,
  Petz2008a, Barndorff-Nielsen2003}, but these works are somewhat
orthogonal to the present treatment because they concern sufficiency
of a quantum system with respect to classical measurement data
\cite{Jencova2006, Petz2008a}, or the sufficiency of measurement data
with respect to preparation data \cite{Barndorff-Nielsen2003}.

\subsubsection{Classical sufficient statistics}

Suppose a parameter, represented by a random variable $Y$, is to be
estimated from data, represented by a random variable $X$.
\begin{definition}
  A \emph{sufficient statistic} for $X$ with respect to $Y$ is a
  function $t$ of the values of $X$ such that the random variable
  $t(X)$ satisfies
  \begin{equation}
    \label{eq:Suff:CDef}
    P(Y|t(X)=t(x)) = P(Y|X=x),
  \end{equation}
  for all $x$ such that $P(X = x) \neq 0$.
\end{definition}
A sufficient statistic for $X$ is a way of processing $X$ such that
the result is just as informative about $Y$ as $X$ is.  In other
words, learning the value of the processed variable $t(X)$ allows an
agent to make all the same inferences about $Y$ that they could have
made by learning the value of $X$ itself.  Such processings are
coarse-grainings of the values of $X$, which discard information about
$X$, but only information that is not relevant for making inferences
about $Y$.

Since $t(X)$ is just a function of $X$, it is immediate that $Y$ is
conditionally independent of $t(X)$ given $X$, i.e.
\begin{equation}
  P(Y|X,t(X)) = P(Y|X),
\end{equation}
This follows from the fact that we can write the joint distribution as $P(Y,X,t(X))=P(t(X)|X)P(YX)$ (where $P(t(X)=a|X=x)=\delta_{a,t(x)}$).
Moreover, the sufficiency condition, eq.~\eqref{eq:Suff:CDef}, implies
that it is also true that $Y$ is conditionally independent of $X$
given $t(X)$, i.e.
\begin{equation}
  P(Y|X,t(X)) = P(Y|t(X)).
\end{equation}
This is a consequence of the fact that the joint distribution can also be written as $P(Y,X,t(X))=P(t(X)|X)P(Y|t(X))P(X)$, where we have used eq.~\eqref{eq:Suff:CDef}.
\begin{definition}
  A \emph{minimal} sufficient statistic for $X$ with respect to $Y$ is
  a sufficient statistic that can be written as a function of any
  other sufficient statistic for $X$ with respect to $Y$.
\end{definition}
A minimal sufficient statistic for $X$ with respect to $Y$ contains
only that information about $X$ that is relevant for making inferences
about $Y$.  Clearly, a sufficient statistic $t$ is minimal iff
\begin{equation}
  \label{eq:Suff:Equiv}
  t(x) = t(x') \,\, \Leftrightarrow \,\, P(Y|X=x) = P(Y|X=x').
\end{equation}

The following lemma is used repeatedly in our discussion of combining
quantum states.

\begin{lemma}
  \label{lem:Suff:BayesianPP}
  Let $P(X,Y)$ be a probability distribution over two random variables
  and let $t(x) = P(Y|X=x)$, i.e. $t$ is a statistic for $X$ that
  takes functions of $Y$ for its values.  Then, $t$ is a minimal
  sufficient statistic for $X$ with respect to $Y$ and
  \begin{equation}
    \label{eq:Suff:BayesianPP}
    P(Y|t(X) = t(x)) = t(x).
  \end{equation}
\end{lemma}

\begin{proof}
  Clearly $t$ satisfies eq.~\eqref{eq:Suff:Equiv} because $t(x)$ is
  \emph{equal} to $P(Y|X=x)$ in this case.  It is therefore minimally
  sufficient.  By the conditional version of belief propagation
  \begin{multline}
    P(Y|t(X) = t(x)) = \\ \sum_{x'} P(Y|X=x',t(X) = t(x))P(X=x'|t(X) = t(x)).
  \end{multline}
  Since $t$ is a sufficient statistic, $A$ is conditionally
  independent of $t(X)$ given $X$, so this reduces to
  \begin{multline}
    P(Y|t(X) = t(x)) = \\ \sum_{x'} P(Y|X=x')P(X=x'|t(X) = t(x)).
  \end{multline}
  The term $P(X=x'|t(X)=t(x))$ is only nonzero for those values $x'$
  such that $t(x')=t(x)$ and all such values satisfy
  $P(Y|X=x')=P(Y|X=x)$.  Therefore,
  \begin{multline}
    P(Y|t(X)=t(x)) = \\ P(Y|X=x) \sum_{\{x'|t(x')=t(x)\}}P(X=x'|t(X)=t(x)).
  \end{multline}
  However, $\sum_{\{x'|t(x')=t(x)\}}P(X=x'|t(X)=t(x)) = \sum_{x'}
  P(X=x'|t(X)=t(x)) = 1$, since $P(X=x'|t(X)=t(x))$ is zero when
  $t(x') \neq t(x)$ and it is a conditional probability distribution.
  Hence,
  \begin{align}
    P(Y|t(X)=t(x)) & = P(Y|X=x)\\
    & = t(x).
  \end{align}
\end{proof}

Eq.~\eqref{eq:Suff:BayesianPP} looks superficially similar to Lewis'
Principal Principle \cite{Lewis1986}, which states that when you know
that the objective chance of an event takes a particular value then
you should assign that value as your subjective probability for that
event.  However, eq.~\eqref{eq:Suff:BayesianPP} is not a statement
about objective chances.  Its interpretation is entirely in terms of
subjective probabilities.  Suppose $P(X,Y)$ is your subjective
probability distribution for $X$ and $Y$ and you announce this to me.
I then go and observe $X$, finding that it has the value $x$.  If I
then tell you that the subjective probability distribution that you
would assign to $Y$ if you knew the value of $X$ that I have observed
is $Q(Y)$, and you believe that I am being honest, i.e.\ that I have
computed $Q(Y) = P(Y|X=x)$ from your subjective probability
distribution and this is what I am reporting back to you, then you
have learned that $t(X) = Q$ and eq.~\eqref{eq:Suff:BayesianPP} says
that your posterior probability distribution for $Y$ should now be
$Q(Y)$.

\subsubsection{Hybrid sufficient statistics}

Recall that if $XB$ is a hybrid region then conditional density
operators $\rho_{B|X}$ are of the form
\begin{equation}
  \rho_{B|X} = \sum_x \Ket{x}\Bra{x}_X \otimes \rho_{B|X=x},
\end{equation}
where the operators $\rho_{B|X=x}$ are normalized density operators on
$\Hilb[B]$.  As in the classical case, the idea of sufficiency is to
find a statistic for $X$ with fewer values than $X$ that still allows
the conditional density operator to be reconstructed.  In order to do
this, it is only necessary to know which density operator
$\rho_{B|X=x}$ a value of $X$ corresponds to, and there may be fewer
distinct density operators than values of $X$.  This motivates the
following definition.
\begin{definition}
  A \emph{sufficient statistic} for $X$ with respect to the quantum
  region $B$ is a function $t$ of the values of $X$ such that the
  random variable $t(X)$ satisfies
  \begin{equation} \label{eq:Qsuffstat}
    \rho_{B|t(X)=t(x)} = \rho_{B|X=x},
  \end{equation}
  for all $x$ such that $\rho_{X = x} \neq 0$.
\end{definition}

This definition captures the notion that learning the value of the
processed variable $t(X)$ allows an agent to make all the same
inferences about the quantum region $B$ that they could have made by
learning the value of $X$ itself.

Since $t(X)$ is just a classical processing of $X$ (specifically,
$\rho_{t(X) = a|X = x} = \delta_{a,t(x)}$), we can introduce a joint
state on the composite system $BXt(X)$ as discussed in
\S\ref{Model:Two} via
\begin{equation}
  \rho_{BXt(X)} = \rho_{t(X)|X} \rho_{BX},
\end{equation}
As one can easily verify, this state satisfies the analogous
conditional independence relations to those that hold in the classical
case.  Specifically, $B$ and $t(X)$ are conditionally independent
given $X$,
\begin{equation}
  \rho_{B|Xt(X)} = \rho_{B|X},
\end{equation}
and because $t(X)$ is a sufficient statistic for $X$ with respect to
$B$, it is also the case that $B$ and $X$ are conditionally
independent given $t(X)$,
\begin{equation}
  \rho_{B|Xt(X)} = \rho_{B|t(X)},
\end{equation}
as can be seen by noting that the joint state can also be written as
$\rho_{BXt(X)} = \rho_{t(X)|X} \rho_{B|t(X)} \rho_X$ if one makes use
of eq.~\eqref{eq:Qsuffstat}.

\begin{definition}
  A \emph{minimal} sufficient statistic for $X$ with respect to a
  quantum region $B$ is a sufficient statistic that can be written as
  a function of any other sufficient statistic for $X$ with respect to
  a quantum region $B$.
\end{definition}
It follows that minimal sufficiency is equivalent to
\begin{equation}
  \label{eq:Suff:QEquiv}
  t(x) = t(x') \,\, \Leftrightarrow \rho_{B|X=x} = \rho_{B|X=x'}.
\end{equation}

We will also need an analog of lemma~\ref{lem:Suff:BayesianPP}.

\begin{lemma}
  \label{lem:Suff:QBayesianPP}
  Let $\rho_{XB}$ be the state of a hybrid region $XB$ and let $t(x) =
  \rho_{B|X=x}$, i.e. $t$ is a statistic for $X$ that takes quantum
  states on $B$ for its values.  Then, $t$ is a minimal sufficient
  statistic for $X$ with respect to $B$ and
  \begin{equation}
    \label{eq:Suff:QBayesianPP}
    \rho_{B|t(X)=t(x)} = t(x).
  \end{equation}
\end{lemma}

\begin{proof}
  The statistic $t$ satisfies eq.~\eqref{eq:Suff:QEquiv} because
  $t(x)$ is \emph{equal} to $\rho_{B|X=x}$.  It is therefore minimally
  sufficient.  By the conditional version of belief propagation
  \begin{equation}
    \rho_{B|t(X) = t(x)} = \Tr[X]{\rho_{B|X t(X)=x}\rho_{X|t(X) = t(x)}}.
  \end{equation}
  Since $t$ is a sufficient statistic, $B$ is conditionally
  independent of $t(X)$ given $X$, so this reduces to
  \begin{equation}
    \rho_{B|t(X) = t(x)} = \Tr[X]{\rho_{B|X}\rho_{X|t(X) = t(x)}}.
  \end{equation}
  However, $\rho_{X=x'|t(X)=t(x)}$ is only nonzero for those values $x'$
  such that $t(x')=t(x)$ and all such values satisfy
  $\rho_{B|X=x'}=\rho_{B|X=x}$.  Therefore,
  \begin{multline}
    \rho_{B|t(X)=t(x)} = \\ \rho_{B|X=x} \sum_{\{x'|t(x')=t(x)\}}\rho_{X=x'|t(X)=t(x)}.
  \end{multline}
  However, $\sum_{\{x'|t(x')=t(x)\}}\rho_{X=x'|t(X)=t(x)} =
  \Tr[X]{\rho_{X|t(X)=t(x)}} = 1$, since $\rho_{X=x'|t(X)=t(x)}$ is
  zero when $t(x') \neq t(x)$ and $\rho_{X|t(X)}$ is a conditional
  state.  Hence,
  \begin{align}
    \rho_{B|t(X)=t(x)} & = \rho_{B|X=x}\\
    & = t(x).
  \end{align}
\end{proof}

\section{Quantum state improvement}

\label{Improve}

State improvement is the task of updating your state assignment in the
light of learning another agent's state assignment.  It is the
simplest example of a procedure for combining different states.  We
adopt the approach of treating the other agent's state assignment as
data and conditioning on it.  In the classical case, this idea is
usually attributed to Morris \cite{Morris1974}.

\subsection{General methodology for state improvement}

Classically, suppose a decision maker, \0, assigns a prior state
$P_0(Y)$ to the variable of interest, $Y$.  \0 may have little or no
specialist knowledge about $Y$, in which case her prior would be
something like a uniform distribution.  In order to improve the
quality of her decision, she consults an expert, \1, who reports her
opinion in the form of a state $Q_1(Y)$.  Assuming that \0 does not
have the expertise to assess the data and arguments by which \1
arrived at her state assignment, the summary $Q_1(Y)$ is all she has
to go on.

In order to improve her state assignment by Bayesian conditioning, \0
has to treat \1's state assignment as data.  This means that she has
to construct a likelihood function $P_0(R|Y)$, where $R$ is a random
variable that ranges over all the possible state assignments that \1
might report.  Since $R$ ranges over a space of functions, there may
be technical difficulties in defining a sample space for it, but in
practice $R$ can usually be confined to well parameterized families of
states, e.g. Gaussian states or a finite set of choices. In assigning
her likelihood function, \0 has to take into account factors such as
\1's trustworthiness, her accuracy in making previous predictions, and
so forth.  Assuming that \0 can do this, she can then update her prior
state via Bayes' theorem to obtain
\begin{equation}
  \label{eq:Improve:Bayes}
  P_0(Y|R = Q_1) = \frac{P_0(R = Q_1 | Y)
    P_0(Y)}{P_0(R = Q_1)}.
\end{equation}
where $P_0(R = Q_1)= \sum_Y P_0(R = Q_1|Y)P_0(Y)$.

Turning to the quantum case, the situation is precisely the same
except that we are now dealing with hybrid regions and the quantum
Bayes' theorem.  Specifically, \0 is now interested in a quantum
region $B$, to which she assigns a prior state $\rho^{(0)}_B$, and \1
announces her expert state assignment $\sigma^{(1)}_B$.  \0 treats
\1's announcement as data and constructs a classical random variable
$R$ that takes \1's possible state assignments as values.
Constructing a sample space for all possible states is again
technically subtle, but in practice attention can be restricted to
well-parameterized families.  \0's likelihood is now a hybrid
conditional state $\rho^{(0)}_{R|B}$ and she updates her prior state
assignment via the hybrid Bayes' theorem to give
\begin{equation}
  \label{eq:Improve:QBayes}
  \rho^{(0)}_{B|R = \sigma^{(1)}_B} = \rho^{(0)}_{R =
    \sigma^{(1)}_B|B} \Sprod \\ \left ( \rho^{(0)}_B
    \left [\rho^{(0)}_{R = \sigma^{(1)}_B}\right ]^{-1}
  \right ).
\end{equation}
where $\rho^{(0)}_{R = \sigma^{(1)}_B} = \Tr[B]{\rho^{(0)}_{R =
    \sigma^{(1)}_B|B} \rho^{(0)}_B}$.

Note that the same methodology can be applied when \0 consults more
than one expert: \1, \2, etc.  \0 simply has to construct a likelihood
function $P(R_1,R_2,\ldots|Y)$ in the classical case or a likelihood
operator $\rho_{R_1R_2,\ldots|B}$ in the quantum case, where $R_1$
represents \1's state assignment, $R_2$ represents \2's state
assignment, etc.  She then applies the appropriate version of Bayes'
theorem to condition on the state assignments that the experts
announce.  This procedure is used in our approach to the pooling
problem, discussed in \S\ref{Pool}.

\subsection{The case of shared priors}

Eqs.~\eqref{eq:Improve:Bayes} and \eqref{eq:Improve:QBayes} are the
general rules that \0 should use to improve her state assignment, but
in practice it can be difficult to determine the likelihoods for $R$
needed to apply them.  However, the rules can simplify drastically in
some situations.  In particular, if \0 and \1 started with a shared
prior for $Y$ or $B$, \1's state differs from \0's due to having
collected more data, and \0 is willing to trust \1's data analysis,
then the rules imply that \0 should just adopt \1's state assignment
wholesale.

Note that, in both the objective and subjective approaches, starting
out with shared priors is an idealization.  In the objective approach
this is because it is unlikely that \0 and \1 have exactly the same
knowledge about the region of interest, and in the subjective approach
this is because their prior beliefs might simply be different.
Nevertheless, in the objective approach we can always imagine a
(possibly hypothetical) time in the past at which \0 and \1 had
exactly the same knowledge and, provided \0's knowledge is a subset of
\1's current knowledge, the result still follows.  This argument does
not apply in the subjective case, but there are still circumstances in
which the ideal of shared priors is a good approximation.

Consider first the classical case.  \0 and \1 share a prior state
assignment $P_0(Y) =P_1(Y) = P(Y)$ for the variable of interest.  \1
then obtains some extra data in the form of the value $x$ of some
random variable $X$ that is correlated with $Y$.  Before learning the
value of $X$, \1 adopts a likelihood model for it, given by
conditional probabilities $P(X|Y)$, and we assume that \0 agrees with
this likelihood model.  Upon acquiring the value $x$ of $X$, \1
updates her probabilities to $Q_1(Y) = P(Y|X=x)$, which can be
computed from Bayes' theorem, and then she reports $Q_1(Y)$ to \0. In
other words, \0 learns that $R = Q_1$ and she must condition on this
data to obtain her improved state assignment $P(Y|R=Q_1)$.

\begin{proposition}
  \label{thm:Improve:Obj}
  If \0 and \1 share a prior state assignment $P(Y)$ and likelihood
  model $P(X|Y)$ for the data collected by \1, then \0's improved
  state is $P(Y|R=Q_1) = Q_1(Y)$, where $Q_1(Y)$ is \1's updated state
  assignment.
\end{proposition}

\begin{proof}
  Because \0 and \1 have a shared prior and likelihood assignment, the
  variable $R$ is simply $R(x) = P(Y|X=x)$, where $P(Y|X=x)$ is the
  probability distribution that \0 would assign if she knew the value
  of $X$.  Lemma~\ref{lem:Suff:BayesianPP} then implies that
  $P(Y|R=Q_1) = Q_1$.
\end{proof}

Note that Aumann \cite{Aumann1976} has argued that there is a unique
posterior that objective Bayesians ought to assign when their state
assignments are common knowledge.  The above theorem is a special case
of this in which the unique state can be easily computed.

In the quantum case, the argument proceeds in precise analogy.  \0 and
\1 start with a shared prior state $\rho_B$ for region $B$.  \1
announces her state assignment $\sigma_B^{(1)}$, which can be
represented as the result of conditioning $B$ on the value $x$ of
a random variable $X$, i.e. $\sigma_B^{(1)} = \rho_{B|X=x}$.  We
assume that \0 and \1 agree upon the likelihood operator $\rho_{X|B}$
for $X$.  \0 then has to compute her improved state
$\rho_{B|R=\sigma_B^{(1)}}$.

\begin{proposition}
  If \0 and \1 share a prior state assignment $\rho_B$ and likelihood
  operator $\rho_{X|B}$ for the data collected by \1, then \0's
  improved state is $\rho_{B|R=\sigma_B^{(1)}} = \sigma_B^{(1)}$,
  where $\sigma_B^{(1)}$ is \1's updated state assignment.
\end{proposition}

The proof is just the obvious generalization of the proof of
theorem~\ref{thm:Improve:Obj}, making use of
lemma~\ref{lem:Suff:QBayesianPP} instead of lemma~\ref{lem:Suff:BayesianPP}.

\subsection{Discussion}

Although our results show that state improvement is trivial in the
case of shared priors, eqs.~\eqref{eq:Improve:Bayes} and
\eqref{eq:Improve:QBayes} are still applicable when \0 and \1 do not
share prior states and, in that case, they give nontrivial results.
The analysis of such cases is a lot more involved, so we do not
consider any examples here.

In the classical case, the general methodology leading to
eq.~\eqref{eq:Improve:Bayes} can be criticized.  It is an onerous
requirement for \0 to be able to articulate a likelihood for all
possible state assignments that \1 might make.  This criticism is
mitigated by the shared priors result, which shows that, at least in
this case, the likelihood model need not be specified in detail.  Such
simplifications might also occur in other models that do not depend on
shared priors.  In any case, this criticism is not particularly unique
to state improvement, since it can be leveled at Bayesian methodology
in general.  It is always a heavy requirement for an agent to specify
a full probability distribution over all the variables of interest.
For this reason, alternative Bayesian theories have been developed
with less onerous requirements, such as the requirement to specify
expectation values rather than full probability distributions
\cite{Finetti1974, Finetti1975, Goldstein1981}.

A criticism that is more specific to state improvement is that the
beliefs that \0 uses to determine $P_0(Y)$ might be correlated with
the beliefs that she uses to determine the likelihood $P_0(R|Y)$,
e.g. \0 might be biased towards believing that \1 will report states
that are concentrated on values of $Y$ that \0 herself believes are
likely.  A generalization that takes these correlations into account
has been proposed \cite{French1980}.

Every criticism leveled against the classical methodology also applies
to the quantum case and, no doubt, the proposed classical
generalizations could be raised to the quantum level by applying the
methods outlined in this paper.  This is not done here because it is
not our goal to say the final word on quantum state improvement, but
only to point out that there is no need to reinvent the wheel when
studying the quantum case because classical methods can be easily
adapted using the formalism of conditional states.

Finally, note that quantum state improvement has previously been
considered by Herbut \cite{Herbut2004}, who adopted an ad hoc
procedure based on closeness of \0 and \1's states with respect to
Hilbert-Schmidt distance.  It would be interesting to see if Herbut's
rule can be derived using Bayesian methodology under a set of
reasonable assumptions that \0 could make about how \1 arrived at her
state assignment.

\section{Quantum state pooling}

\label{Pool}

The problem of state pooling concerns what happens when agents who
each have their own state assignments want to make decisions as a
group.  To do so, they need to come up with a state assignment that
accurately reflects the views of the group as a whole.

In an ideal world, the agents would first reconcile their differences
empirically so that everyone agrees on a common state assignment.  The
discussion of subjective Bayesian compatibility shows that it is
possible for this to happen if their states satisfy the BFM
compatibility criterion.  Furthermore, as a consequence of the
classical and quantum de Finetti theorems \cite{Finetti1937,
  Finetti1975, Hudson1976, Caves2002b}, if the agents can construct an
exchangeable sequence of experiments then their states can be expected
to converge in the long run by application of Bayesian conditioning.
Nevertheless, it is not always possible to collect more data before a
decision has to be made and, for the subjective Bayesian, there is
also the question of how to combine sharply contradictory beliefs that
do not satisfy compatibility criteria in the first place.

The goal of this section is to provide a general methodology for
quantum pooling based on applying the principles of quantum Bayesian
inference, similar to the approach to state improvement developed in
\S\ref{Improve}.  In the case of shared priors, we also derive a
specific pooling rule from this methodology that was previously
proposed by Spekkens and Wiseman \cite{Spekkens2007}.  However, before
embarking upon this discussion, it is useful to take a step back and
look at the basic requirements for pooling and some of the specific
pooling rules that have been proposed in the classical case.

\subsection{Review of pooling rules}

One reasonable requirement for a pooling rule is that the pooled state
should be compatible with each agent's individual state assignment.
If this is so then each agent is assured that it is possible for them
to be vindicated by future observations.  This is
because subjective Bayesian compatibility guarantees that, for each
agent, it is possible that data could be collected that would cause
the pooled state and the agent's individual state to become identical
upon Bayesian conditioning.

Consider the classical case where $n$ agents assign states $Q_1(Y),
Q_2(Y), \ldots Q_n(Y)$ to a random variable $Y$.  A \emph{linear
  opinion pool} is a rule where the pooled state $Q_{\text{lin}}$ is
of the form
\begin{equation}
  \label{eq:Pool:Lin}
  Q_{\text{lin}}(Y) = \sum_{j=1}^n w_j Q_j(Y),
\end{equation}
where $0 < w_j < 1$ and $\sum_{j=1}^n w_j = 1$.  The weight $w_j$ can
be thought of as a measure of the amount of trust that the group
assigns to the $j$th agent.  The state $Q_{\text{lin}}(Y)$ is BFM
compatible with every $Q_j(Y)$ because eq.~\eqref{eq:Pool:Lin} is an
ensemble decomposition of $Q_{\text{lin}}(Y)$ in which each agent's
state appears.  A linear opinion pool is typically less sharply peaked
than the individual agents' assignments. In particular its entropy
cannot be lower than that of the lowest entropy individual state.  It
may be appropriate to use it as a diplomatic solution.  Indeed, this
sort of pooling rule may be applied even if the agents' state
assignments are not pairwise compatible.

Linear opinion pools can be straightforwardly generalized to the
quantum case.  Specifically, if $n$ agents assign states
$\sigma^{(1)}_B, \sigma^{(2)}_B, \ldots \sigma^{(n)}_B$ to a quantum
region $B$, then a quantum linear opinion pool is a rule where the
pooled state $\sigma^{(\text{lin})}_B$ is of the form
\begin{equation}
  \sigma^{(\text{lin})}_B = \sum_{j=1}^n w_j \sigma^{(j)}_B,
\end{equation}
where $0 < w_j < 1$ and $\sum_{j=1}^n w_j = 1$.  Similar remarks apply
to this as to the classical case.

Classically a \emph{multiplicative opinion pool}\footnote{The usual
  terminology for this is a \emph{logarithmic} opinion pool, since it
  corresponds to a linear rule for combining the logarithms of states.
  However, we prefer the term multiplicative because log-linearity no
  longer holds in the quantum generalization.} is a rule whereby the
pooled state is of the form
\begin{equation}
  \label{eq:Pool:Mult}
  Q_{\text{mult}}(Y) = c \prod_{j=1}^n Q_j(Y)^{w_j},
\end{equation}
where $c$ is a normalization constant,
\begin{equation}
  c = \frac{1}{\sum_Y  \prod_{j=1}^n Q_j(Y)^{w_j}}.
\end{equation}
Multiplicative pools typically result in a pooled state
that is more sharply peaked than any of the individual agent's
states. Normalizability implies that multiplicative pools can only be
applied to states that are jointly compatible, meaning that there is
at least one value $y$ of $Y$ such that $Q_j(Y=y) > 0$ for all $j$.
Any such value has nonzero weight in $Q_{\text{mult}}(Y)$, which
guarantees that $Q_{\text{mult}}(Y)$ is compatible with every agent's
individual assignment.  As shown below, a multiplicative pool may be
appropriate in an objective Bayesian framework where all the agents
start with a shared uniform prior and the differences in their state
assignments result from having collected different data.  

In order to
account for the case where the shared prior is not uniform, the
multiplicative pool has to be generalized to
\begin{equation}
  \label{eq:Pool:GMult}
  Q_{\text{gmult}}(Y) = c \prod_{j=0}^n Q_j(Y)^{w_j},
\end{equation}
where the extra state $Q_0(Y)$ represents the shared prior
information.

Unlike with linear pools, it is not immediately obvious how to
generalize multiplicative pools to the quantum case because the
product of states in eq.~\eqref{eq:Pool:GMult} does not have a unique
generalization due to non-commutativity.

\subsection{General methodology for state pooling}

As with the other problems tackled in this paper, pooling rules should
be derived in a principled way from the rules of Bayesian inference,
rather than simply being posited.  One way to do this to adopt the
\emph{supra-Bayesian} approach.  This works by requiring the group of
agents to put themselves in the shoes of \0 the decision maker who we
met in the state improvement section.  Specifically, in the classical
case, acting together, they are asked to come up with a likelihood
function $P_0(R_1,R_2,\ldots,R_n|Y)$ that they think a neutral
decision maker (\0 the supra-Bayesian) would assign, where $R_j$ is a
random variable that ranges over all possible state assignments that
the $j$th agent might make; and a prior $P_0(Y)$, which can often just
be taken to be the uniform distribution or a shared prior that the
agents may have agreed upon at some point in the past before their
opinions diverged.  They can then update $P_0(Y)$ to
$P_0(Y|R_1=Q_1,R_2=Q_2,\ldots,R_n=Q_n)$ via Bayesian conditioning and
use this as the pooled state $Q_{\text{supra}}(Y)$.  Pooling then
becomes just an application of the state improvement method discussed
in the previous section.  In the quantum case, the equivalent
ingredients are a hybrid likelihood $\rho^{(0)}_{R_1R_2\ldots R_n|B}$
and a prior quantum state $\rho^{(0)}_B$, and then the pooled state is
$\sigma^{(\text{supra})}_B =
\rho^{(0)}_{B|R_1=\sigma^{(1)}_B,R_2=\sigma^{(2)}_B,\ldots,R_n=\sigma^{(n)}_B}$,
which can be computed from quantum Bayesian conditioning.

Admittedly, it might be a pretty tall order to expect the agents to be
able to act together as a fictional supra-Bayesian \0, but this method
does allow conditions under which the different pooling rules should
be used to be derived rigorously, which in turn gives insight into
when they might be useful as rules-of-thumb more generally.  It also
has the advantage that it allows quantum generalizations to be derived
unambiguously, since the necessary tools of quantum Bayesian inference
have been developed in \cite{Leifer2011a} and the preceding sections.
In particular, it resolves the ambiguity surrounding the correct
quantum generalization of the multiplicative opinion pool.

To illustrate this, we show that, in the case of shared priors, the
supra-Bayesian approach can be used to motivate the two-agent case of
the quantum generalized multiplicative pool with $w_0 = -1, w_1 = 1,
w_2 = 1$.

\subsection{The case of shared priors}

For simplicity, we specialize to the case of a group of two agents, \1
and \2.  First consider the classical case where \1 and \2 have
individual state assignments $Q_1(Y)$ and $Q_2(Y)$.  We assume that \1
and \2 started from a shared prior $P(Y)$, which can be used as
\0's prior $P_0(Y) = P(Y)$ in the supra-Bayesian approach, and that
the current differences in \1 and \2's state assignments are due to
having collected different data. The additional data available to \1
and \2 are modeled as values $x_1$ and $x_2$ of random variables $X_1$
and $X_2$ respectively.  Before learning the values of $X_1$ and
$X_2$, \1 and \2 assigned likelihood functions $P(X_1|Y)$ and
$P(X_2|Y)$, which, when combined with the prior $P(Y)$, determine
their current state assignments via Bayes' theorem, i.e. $Q_j(Y) =
P(Y|X_j=x_j)$.

We assume that it is possible to assign a joint likelihood function
$P(X_1,X_2|Y)$, such that $P(X_1|Y)$ and $P(X_2|Y)$ are obtained by
marginalization.  It is unrealistic to think that \1 and \2 must
specify this joint likelihood in detail.  Fortunately, in order to obtain
a generalized multiplicative pool, they need only agree on some of its
broad features.  In particular, if they agree that minimal sufficient
statistics for $X_1$ and $X_2$ are conditionally independent given $Y$, then
supra-Bayesian pooling gives rise to a generalized multiplicative
pool.

\begin{theorem}
  \label{thm:Pool:Obj}
  If a minimal sufficient statistic for $X_1$ with respect to $Y$ and
  a minimal sufficient statistic for $X_2$ with respect to $Y$ are
  conditionally independent given $Y$, then the supra-Bayesian pooled state
  $Q_{\emph{supra}}(Y) = P_0(Y|R_1=Q_1,R_2=Q_2)$ is given by
  \begin{equation}
    Q_{\emph{supra}}(Y) = c \frac{Q_1(Y) Q_2(Y)}{P(Y)},
  \end{equation}
  where $c$ is a normalization factor, independent of $Y$.
\end{theorem}

Comparing this result with eq. \eqref{eq:Pool:GMult} shows that this
is a generalized multiplicative pool with $Q_0(Y)=P(Y)$, $w_0=-1$,
$w_1=1$ and $w_2 = 1$.  In the special case of a uniform prior, this
reduces to
\begin{equation}
  Q_{\text{supra}}(Y) = c' Q_1(Y) Q_2(Y),
\end{equation}
where $c'$ is a different normalization constant.  This is a
multiplicative pool with $w_1=1$ and $w_2=1$.

\begin{proof}[Proof of theorem~\ref{thm:Pool:Obj}]
  By definition, the supra-Bayesian pooled state is
  $Q_{\text{supra}}(Y) = P(Y|R_1=Q_1,R_2=Q_2)$ and this can be
  computed from the prior $P(Y)$ and the likelihood $P(R_1,R_2|Y)$ via
  Bayes' theorem.  Now, $R_j$ can be thought of as a function-valued
  statistic for $X_j$ via $R_j(x_j) = P(Y|X_j=x_j)$.  It is a minimal
  sufficient statistic with respect to $Y$ because $R_j(x_j) =
  R_j(x_j')$ iff $P(Y|X_j=x_j) = P(Y|X_j=x_j')$.  By assumption, there
  exist minimal sufficient statistics for $X_1$ and for $X_2$ that are
  conditionally independent given $Y$.  However, any minimal
  sufficient statistic is a bijective function of any other minimal
  sufficient statistic for the same variable, so if any pair of such
  statistics are conditionally independent then they all are.
  Therefore, $R_1$ and $R_2$ are conditionally independent given $Y$,
  and so by \ref{item:CI:Def4}:
  \begin{equation}
    P(R_1,R_2|Y) = P(R_1|Y)P(R_2|Y).
  \end{equation}
  The terms $P(R_j|Y)$ can be inverted via Bayes' theorem to obtain
  $P(R_j|Y) = P(Y|R_j)P(R_j)/P(Y)$, which gives
  \begin{equation}
    P(R_1,R_2|Y) = P(R_1)P(R_2) \frac{P(Y|R_1)P(Y|R_2)}{P(Y)^2}.
  \end{equation}
  Using Bayes' theorem again in the form $P(Y|R_1,R_2) =
  P(R_1,R_2|Y)P(Y)/P(R_1,R_2)$ gives
  \begin{equation}
    P(Y|R_1,R_2) =
    \frac{P(R_1)P(R_2)}{P(R_1,R_2)}\frac{P(Y|R_1)P(Y|R_2)}{P(Y)},
  \end{equation}
  which, upon substituting the announced values of $R_1$ and $R_2$,
  gives
  \begin{multline}
    Q_{\text{supra}}(Y) =
    \frac{P(R_1=Q_1)P(R_2=Q_2)}{P(R_1=Q_1,R_2=Q_2)} \\
    \times \frac{P(Y|R_1=Q_1)P(Y|R_2=Q_2)}{P(Y)}.
  \end{multline}
  The term $c = \left [ P(R_1=Q_1)P(R_2=Q_2) \right
  ]/P(R_1=Q_1,R_2=Q_2)$ is independent of $Y$, so it can be determined
  from the normalization constraint $\sum_Y Q_{\text{supra}}(Y)=1$.
  Also, lemma~\ref{lem:Suff:BayesianPP} implies $P(Y|R_j=Q_j) =
  Q_j(Y)$, so we have
  \begin{equation}
    Q_{\text{supra}}(Y) = c \frac{Q_1(Y)Q_2(Y)}{P(Y)},
  \end{equation}
  as required.
\end{proof}

In the quantum case, \1 and \2 have individual state assignments
$\sigma^{(1)}_B$ and $\sigma^{(2)}_B$.  Again, any differences in \1
and \2's state assignments are assumed to arise from having collected
different data, before which they agreed upon a shared prior $\rho_B$,
which can be used as \0's prior state $\rho^{(0)}_B = \rho_B$ in the
supra-Bayesian approach.

Again, we assume that \1 and \2 have observed values $x_1$ and $x_2$
of random variables $X_1$ and $X_2$, with likelihood operators,
$\rho_{X_1|B}$ and $\rho_{X_2|B}$.  \1 and \2's states result from
conditioning the shared prior on their data using these likelihoods.
We assume that there is a joint likelihood $\rho_{X_1X_2|B}$, of which
\1 and \2's likelihoods are marginals.  \1 and \2 need not agree on
the full details of this joint likelihood, only that minimal
sufficient statistics for $X_1$ and $X_2$ satisfy \ref{item:CI:QDef4},
which is slightly weaker than conditional independence.  We then have

\begin{theorem}
  \label{thm:Pool:QObj}
  If a minimal sufficient statistic $t_1$ for $X_1$ with respect to
  $B$ and a minimal sufficient statistic $t_2$ for $X_2$ with respect
  to $B$ satisfy
  \begin{equation}
    \label{eq:Pool:QCond}
    \rho_{t_1(X_1) t_2(X_2)|B} = \rho_{t_1(X_1)|B} \rho_{t_2(X_2)|B},
  \end{equation}
  then the supra-Bayesian pooled
  state $\sigma^{(\emph{supra})}_B =
  \sigma^{(0)}_{B|R_1=\sigma^{(1)}_B,R_2=\sigma^{(2)}_B}$ is given by
  \begin{equation}
    \label{eq:Pool:QGMult}
    \sigma^{(\emph{supra})}_B = c \sigma^{(1)}_B \rho_{B}^{-1}
    \sigma^{(2)}_B
  \end{equation}
  where $c$ is a normalization factor, independent of $B$.
\end{theorem}

Eq.~\eqref{eq:Pool:QGMult} is the quantum generalization of the
generalized multiplicative pool with $w_0=-1$, $w_1=1$, $w_2=1$.
Despite appearances, this expression is symmetric under exchange of
$1$ and $2$.  This follows from the condition \eqref{eq:Pool:QCond},
which implies that $\rho_{t_1(X_1)|B}$ and $\rho_{t_2(X_2)|B}$ must
commute.  When $\rho_B$ is a maximally mixed state,
eq.~\eqref{eq:Pool:QGMult} reduces to
\begin{equation}
  \sigma^{(\text{supra})}_B = c' \sigma^{(1)}_B \sigma^{(2)}_B,
\end{equation}
where $c'$ is a different normalization constant.  This is a quantum
generalization of the multiplicative pool with $w_1=1$, $w_2=1$.

Although conditional independence of the minimal sufficient statistics
was assumed in the classical case, eq.~\eqref{eq:Pool:QCond} is
strictly weaker than conditional independence, as explained in
\S\ref{CI}.

\begin{proof}[Proof of theorem~\ref{thm:Pool:QObj}]
  By definition, the supra-Bayesian pooled state is
  $\sigma^{(\text{supra})}_B = \rho_{B|R_1 = \sigma_B^{(1)}, R_2 =
    \sigma_B^{(2)}}$ and this can be computed from the prior $\rho_B$
  and the likelihood $\rho_{R_1R_2|B}$ via Bayes' theorem.  Each $R_j$
  is an operator-valued statistic for $X_j$ via $R_j(x_j) =
  \rho_{B|X_j=x_j}$.  They are minimal sufficient statistics with
  respect to $B$ because $R_j(x_j) = R_j(x_j')$ iff $\rho_{B|X_j=x_j}
  = \rho_{B|X_j=x_j'}$.  By assumption, there exist minimal sufficient
  statistics, $t_1$ and $t_2$, for $X_1$ and $X_2$ that satisfy
  \begin{equation}
    \rho_{t_1(X_1) t_2(X_2)|B} = \rho_{t_1(X_1)|B} \rho_{t_2(X_2)|B},
  \end{equation}
  but since any minimal sufficient statistic is a bijective function
  of any other minimal sufficient statistic for the same variable,
  $R_1$ and $R_2$ must also satisfy
  \begin{equation}
    \rho_{R_1R_2|B} = \rho_{R_1|B} \rho_{R_2|B}.
  \end{equation}
  The terms $\rho_{R_j|B}$ can be inverted via Bayes' theorem to
  obtain $\rho_{R_j|B} = \rho_{B|R_j} \Sprod \left ( \rho_{R_j}
    \rho_B^{-1} \right )$, which gives
  \begin{equation}
    \rho_{R_1R_2|B} = \left [  \rho_{B|R_1} \Sprod \left ( \rho_{R_1}
        \rho_B^{-1} \right )\right ] \left [  \rho_{B|R_1} \Sprod \left ( \rho_{R_1}
        \rho_B^{-1} \right ) \right ].
  \end{equation}
  Since $R_1$ and $R_2$ are classical, the operators $\rho_{R_j}$
  commute with everything else and so expanding the $\Sprod$-products
  gives
  \begin{equation}
    \rho_{R_1R_2|B} = \rho_{R_1} \rho_{R_2} \rho_{B}^{-\frac{1}{2}}
    \rho_{B|R_1} \rho_B^{-1} \rho_{B|R_2} \rho_B^{-\frac{1}{2}}.
  \end{equation}
  Using Bayes' theorem again in the form $\rho_{B|R_1R_2} =
  \rho_{R_1R_2|B} \Sprod \left ( \rho_{B} \rho_{R_1R_2}^{-1} \right )$
  and noting that $\rho_{R_1R_2}$ commutes with everything else gives
  \begin{equation}
    \rho_{B|R_1R_2} = \rho_{R_1}\rho_{R_2}\rho_{R_1R_2}^{-1} \left (
      \rho_{B|R_1} \rho_B^{-1} \rho_{B|R_2}
    \right ),
  \end{equation}
  which, upon substituting the announced values of $R_1$ and $R_2$,
  gives
  \begin{multline}
    \sigma_B^{(\text{supra})} = \frac{\rho_{R_1 = \sigma_B^{(1)}}\rho_{R_2 =
        \sigma_B^{(2)}}}{\rho_{R_1 = \sigma_B^{(1)}, R_2 =
        \sigma_B^{(2)}}} \\ \times
      \rho_{B|R_1 = \sigma_B^{(1)}} \rho_B^{-1} \rho_{B|R_2 = \sigma_B^{(2)}},
  \end{multline}
  The term $c = \left [\rho_{R_1 = \sigma_B^{(1)}}\rho_{R_2 =
      \sigma_B^{(2)}} \right ] / \rho_{R_1 = \sigma_B^{(1)}, R_2 =
    \sigma_B^{(2)}}$ is independent of $B$, so it can be determined
  from the normalization constraint
  $\Tr[B]{\sigma_B^{(\text{supra})}}=1$.  Also,
  lemma~\ref{lem:Suff:QBayesianPP} implies
  $\rho_{B|R_j=\sigma^{(j)}_B} = \sigma^{(j)}_B$, so we have
  \begin{equation}
    \sigma_B^{(\text{supra})} = c \sigma_B^{(1)} \rho_B^{-1} \sigma_B^{(2)},
  \end{equation}
as we set out to prove.
\end{proof}

\subsection{Comparison to other approaches}

Quantum state pooling has been discussed previously in
\cite{Brun2002a, Jacobs2002, Poulin2003, Jacobs2005, Spekkens2007}.
Both \cite{Brun2002a} and \cite{Poulin2003} propose pooling
methodologies that seem ad hoc from the Bayesian point of view, but,
as with Herbut's approach to improvement, it would be interesting to
see whether they could be justified in the supra-Bayesian approach.

Jacobs \cite{Jacobs2002, Jacobs2005} considers quantum state pooling
in the case where \1 and \2 arrive at their states by making direct
measurements on the system of interest.  In particular, he derives a
generalization of the multiplicative rule that is distinct from the
one we derive.  From the perspective of the conditional states
formalism, his rule is not a valid way of combining state assignments.
The reason is that Jacobs takes collapse rules in quantum theory
--- such as the von Neumann-L\"{u}ders-von Neumann projection
postulate or its generalization to POVMs --- as quantum versions of
Bayesian conditioning, but in the conditional states framework, such
collapse rules are explicitly \emph{not} instances of Bayesian
conditioning, as argued in \S\ref{Model:Single} and \cite{Leifer2011a}.

Spekkens and Wiseman \cite{Spekkens2007} consider the case of pooling
via remote measurements, wherein there is a shared prior state
$\rho_{BA_1A_2}$ of a tripartite system and \1 and \2 arrive at their
differing state assignments for $B$ by making POVM measurements on
$A_1$ and $A_2$ respectively, as depicted in
fig.~\ref{fig:Model:Remote}.  They obtain the same generalized
multiplicative pool that has been derived here, namely $c
\sigma^{(1)}_B \rho_B^{-1} \sigma^{(2)}_B$, for two restricted classes
of states $\rho_{BA_1A_2}$.  Both of these classes are special cases
of states for which $A_1$ and $A_2$ are conditionally independent
given $B$.  If $\rho_{BA_1A_2}$ satisfies this conditional
independence then so does any hybrid state $\rho_{BX_1X_2}$ obtained
by measuring POVMs $\rho_{X_1|A_1}$ on system $A_1$ and
$\rho_{X_2|A_2}$ on system $A_2$.  This is because the conditional
mutual information cannot be increased by applying local CPT maps to
$A_1$ and $A_2$.  The minimal sufficient statistics for $X_1$ and
$X_2$ then also satisfy conditional independence because they are just
local processings of $X_1$ and $X_2$.  Therefore, the assumptions of
theorem~\ref{thm:Pool:QObj} follow from this conditional independence.
As such, the result of \cite{Spekkens2007} is seen to be a special
case of the one derived here.

What we have shown is that the Spekkens and Wiseman pooling rule
holds under much weaker conditions than the conditional independence
of $A_1$ and $A_2$ given $B$.  For example, it also holds for states of the form $\rho_{BA_1'A_2'}
\otimes \rho_{A_1''A_2''}$, where $\Hilb[A_j] =
\Hilb[A_j']\otimes\Hilb[A_j'']$ and $A_1'$ and $A_2'$ are
conditionally independent given $B$.  For such states, $A_1$ and $A_2$
are not conditionally independent given $B$ whenever
$\rho_{A_1''A_2''}$ is a correlated state, but $A_1''$ and $A_2''$
contain no information about $B$, so they will not be correlated with
the minimal sufficient statistics for $X_1$ and $X_2$ and consequently
the minimal sufficient statistics \emph{are} conditionally independent
given $B$, which is sufficient to derive the result\footnote{It would
  be interesting to fully classify the set of tripartite states
  $\rho_{BA_1A_2}$ for which the multiplicative pooling rule applies
  for all remote measurements.}.  Of course, our results also
significantly generalize those of \cite{Spekkens2007} because
theorem~\ref{thm:Pool:QObj} applies to a broader set of causal
scenarios than just the remote measurement scenario.

Finally, it is worth pointing out that both Jacobs \cite{Jacobs2002,
  Jacobs2005} and Spekkens and Wiseman \cite{Spekkens2007} adopt a
pooling methodology that is less widely applicable than the one used in the
present work.  In \cite{Spekkens2007}, for example, a fourth party
called Oswald (the overseer) is introduced into the game, in addition
to the two agents and the decision-maker (whom they call the pooler).
Before any data is collected, everyone shares a prior $\rho_B$ for the
region of interest.  In addition, \1, \2 and Oswald assign a shared
prior $\rho_{BX_1X_2}$ including the data variables that \1 and \2 are
going to observe\footnote{Actually, only Oswald needs to know the full
  prior.  \1 and \2 can make do with knowing the reduced states
  $\rho_{BX_1}$ and $\rho_{BX_2}$ respectively.}.  Oswald has access
to both \1 and \2's data, i.e.\ he learns the values $x_1$ and $x_2$
that \1 and \2 observe so he can update his state to the posterior
$\rho_{B|X_1=x_1,X_2=x_2}$.  It is then asserted that if Oswald's
posterior can be determined from the data available to \0, then this
is what she should assign as the pooled state.  Since \0 only knows \1
and \2's state assignments and the prior $\rho_B$, this is possible
only if Oswald's posterior can be computed from these alone.

This methodology is less widely applicable than the one presented here
because it does not specify what to do if \0 cannot determine Oswald's
posterior, whereas ours does. In fact, there are situations in which the multiplicative pooling rule
is applicable even though \0 cannot
determine Oswald's posterior using the data that she has available.
Therefore, even though the rule of adopting Oswald's posterior if it
can be determined is indeed correct in the supra-Bayesian approach,
requiring this is an unnecessary restriction and it is better to make
do without Oswald.

It is useful to consider how such situations can arise.  By learning
$\rho_{B|X_1=x_1}$ and $\rho_{B|X_2=x_2}$, \0 learns a minimal
sufficient statistic for $X_1$ with respect to $B$ and a minimal
sufficient statistic for $X_2$ with respect to $B$ and hence \0's
posterior is $\rho_{B|R_1(X_1) = R_1(x_1),R_2(X_2)=R_2(x_2)}$, where
the function $R_j(x_j) = \rho_{B|X_j=x_j}$ is the state-valued minimal
sufficient statistic for $X_j$.  This is identical to Oswald's
posterior iff $(R_1,R_2)$ happens to be a sufficient statistic for the
pair $(X_1,X_2)$ with respect to $B$, i.e.\ iff
$\rho_{B|R_1(X_1)=R_1(x_1)R_2(X_2) = R_2(x_2)} = \rho_{B|X_1=x_1
  X_2=x_2}$.  In general, this is not the case, since it is only
guaranteed that $R_1$ and $R_2$ are locally sufficient for the
individual data, i.e. $\rho_{B|R_1(X_1) = R_1(x_1)} = \rho_{B|X_1 =
  x_1}$ and $\rho_{B|R_2(X_2) = R_2(x_2)} = \rho_{B|X_2 = x_2}$, and
not globally sufficient for the pair.  However, \0 only has enough
data to reconstruct Oswald's posterior if they are in fact globally
sufficient, that is, if $\rho_{B|R_1(X_1)=R_1(x_1),R_2(X_2)=R_2(x_2)}
= \rho_{B|X_1=x_1,X_2=x_2}$.

\begin{table}[!htb]
  \begin{tabular}{|c|cccccccc|}
    \hline
    \T $Y$ & $0$ & $0$ & $0$ & $0$ & $1$ & $1$ & $1$ & $1$ \\
    \M $X_1$ & $0$ & $0$ & $1$ & $1$ & $0$ & $0$ & $1$ & $1$ \\
    \M\B $X_2$ & $0$ & $1$ & $0$ & $1$ & $0$ & $1$ & $0$ & $1$ \\
    \hline
    \T\B $P(Y,X_1,X_2)$ & $\frac{1}{4}$ & $0$ & $0$ & $\frac{1}{4}$ & $0$ &
    $\frac{1}{4}$ & $\frac{1}{4}$ & $0$ \\
    \hline
  \end{tabular}
  \caption{\label{tbl:App:Pool:LLExa}A prior state for which \0
    cannot determine Oswald's prior, but for which the multiplicative
    pooling rule still holds}
\end{table}

A classical example suffices to show that our pooling rule sometimes
applies even in cases where \0 cannot reconstruct Oswald's
posterior. Suppose $Y$, $X_1$ and $X_2$ are classical bits and
Oswald's prior is given by table~\ref{tbl:App:Pool:LLExa}.  With this
assignment, the shared prior for $Y$ is $P(Y=0) = P(Y=1) =
\frac{1}{2}$.  Learning the value of $X_j$ on its own gives no further
information about $Y$, i.e. $P(Y|X_j = x_j) = P(Y)$, independently of
the value of $X_j$, so both \1 and \2 simply report the uniform
distribution back to \0.  Any minimal sufficient statistic for $X_j$
is trivial, consisting of just a single value, so the sufficient
statistics for $X_1$ and $X_2$ are trivially conditionally independent
and thus our derivation of the multiplicative pooling rule holds.
Unsurprisingly, in this case it just says that \0 should continue to
assign the uniform distribution.  On the other hand, knowing both the
value of $X_1$ and the value of $X_2$ is enough to determine $Y$
uniquely, so Oswald's posterior is a point measure and there is no way
that \0 could determine it from the data she has available.  The
reason why this happens is that all the information about $Y$ is
contained in the correlations between $X_1$ and $X_2$,
i.e. $P(Y=0|X_1=X_2) = 1$ and $P(Y=0|X_1\neq X_2) = 0$, and Oswald is
the only agent who has access to this data.

\section{Conclusions}

In this paper, we have developed a Bayesian approach to quantum state
compatibility, improvement and pooling, based on the principle that
states should always be updated by a quantum analog of Bayesian
conditioning.  This improves upon previous approaches, which were more
ad hoc in nature.  Due to our use of the conditional states formalism,
our results apply to a much wider range of causal scenarios than
previous approaches.  Indeed, the ability of this formalism to unify
the description of many distinct causal arrangements explains the
otherwise puzzling fact that authors considering very different causal
arrangements have found the same results.  For instance, the
compatibility criterion found by Brun, Finkelstein and Mermin in the
case of remote measurements \cite{Brun2002} is identical to the one
found by Jacobs in the case of sequential measurements
\cite{Jacobs2002}.

This paper only represents the beginning of a Bayesian approach to
these problems; there is a lot of scope for further work.  For
example, it would be interesting to determine when a quantum linear
pooling rule can be derived from Bayesian principles, as it has been
in the classical case \cite{Genest1985}, and whether the results of
previous methodologies for quantum state improvement and pooling can
be reconstructed from a Bayesian point of view.  However, perhaps the
most important lesson of this paper is that the conditional states
formalism can vastly simplify the task of generalizing results from
classical probability to the quantum domain.  Definitions, theorems
and proofs can often be ported almost mechanically from classical
probability to quantum theory by making use of the appropriate
analogies.  Many aspects of quantum theory that might appear, by the
lights of the conventional quantum formalism, to have no good
classical analogue, are seen under the new formalism to be
generalizations of very familiar features of Bayesian probability
theory.  As such, this new formalism helps us to focus our attention
on those aspects of quantum theory that \emph{truly} distinguish it
from classical probability theory, such as violations of Bell
inequalities, the impossibility of broadcasting, and monogamy
constraints on correlations.

\begin{acknowledgments}
M.L. would like to thank the UCL Department of Physics and Astronomy for
their hospitality.  M.L. is supported by Grant RFP1-06-006 from The
Foundational Questions Institute.  Part of this work was completed
whilst M.L. was a postdoctoral fellow at the Institute for Quantum
Computing, University of Waterloo, Canada, where he was supported in
part by MITACS and ORDCF, and an associate postdoctoral fellow at the
Perimeter Institute for Theoretical Physics.  Research at Perimeter Institute is
supported in part by the Government of Canada through NSERC and by the
Province of Ontario through MRI.
\end{acknowledgments}

\bibliography{quantumpooling}

\end{document}